\documentclass[journal]{IEEEtran}
\usepackage{graphicx}
\usepackage{amsmath} 
\usepackage{amsthm}
\usepackage{tabularx}
\usepackage{xcolor}
\usepackage{amssymb}
\usepackage{tikz}
\usepackage{booktabs}
\usetikzlibrary{decorations.pathreplacing}  
\usetikzlibrary{shapes.arrows}              
\usetikzlibrary{shadows}                    
\usetikzlibrary{calc}                       
\usetikzlibrary{shapes, arrows.meta, positioning, calc, backgrounds, shadows}
\definecolor{ieee_blue}{RGB}{0, 76, 153}
\definecolor{ieee_orange}{RGB}{255, 127, 0}
\definecolor{ieee_red}{RGB}{180, 20, 20}

\definecolor{tech_red}{RGB}{220, 50, 47}    
\definecolor{tech_blue}{RGB}{38, 139, 210}  
\definecolor{tech_gray}{RGB}{101, 123, 131} 
\definecolor{myblue}{RGB}{0, 114, 189}   
\definecolor{myblack}{RGB}{30, 30, 30}   

\usetikzlibrary{arrows.meta, decorations.pathreplacing}

\usepackage{pgfplots}
\pgfplotsset{compat=1.17}
\usepgfplotslibrary{groupplots} 
\usetikzlibrary{patterns}
\usepackage{hyperref}

\usepackage{tabularx}
\usepackage{booktabs}
\usepackage{array}
\usepackage{graphicx}

\usetikzlibrary{backgrounds}
\pgfdeclarelayer{backgrounds}
\pgfsetlayers{backgrounds,main}

\usepackage{placeins}
\FloatBarrier

\theoremstyle{definition} 
\newtheorem{theorem}{Theorem}

\newtheoremstyle{remarkbf} 
{3pt}   
{3pt}   
{\normalfont} 
{}      
{\bfseries} 
{.}     
{0.5em} 
{}      

\theoremstyle{remarkbf}
\newtheorem{remark}{Remark}

\begin{document}

\title{Wireless Streamlet: A Spectrum-Aware and Cognitive Consensus Protocol for Edge IoT}

\author{Taotao~Wang,~\IEEEmembership{Member,~IEEE,}~Long~Shi,~\IEEEmembership{Senior~Member,~IEEE,}~Fang~Liu,~Qing~Yang,~\IEEEmembership{Member,~IEEE,}~and~Shengli~Zhang,~\IEEEmembership{Senior~Member,~IEEE}
\thanks{T.~Wang, F. Liu, Q. Yang and S. Zhang are with the College of Electronics and Information Engineering, Shenzhen University, Shenzhen 518060, China, e-mail: \{ttwang, liuf, yang.qing, zsl\}@szu.edu.cn.}
\thanks{L. Shi is with the School of Electronic and Optical Engineering, Nanjing University of Science and Technology, Nanjing 210094, China, e-mail: longshi@njust.edu.cn.}

}

\markboth{Wireless Streamlet: A Spectrum-Aware and Cognitive Consensus Protocol for Edge IoT}%
{Shell \MakeLowercase{\textit{et al.}}: Bare Demo of IEEEtran.cls for IEEE Journals}

\maketitle

\begin{abstract}
Blockchain offers a decentralized trust framework for the Internet of Things (IoT), yet deploying consensus in spectrum-congested and dynamic wireless edge IoT networks faces fundamental obstacles: traditional BFT protocols are spectrum-ignorant, leading to inefficient resource utilization and fragile progress under time-varying interference. This paper presents \textit{Wireless Streamlet}, a spectrum-aware and cognitive consensus protocol tailored for wireless edge IoT. Building on Streamlet's streamlined structure, we introduce a \textit{Channel-Aware Leader Election (CALE)} mechanism. CALE serves as a verifiable cross-layer cognitive engine that leverages receiver-measured channel state information (CSI) piggybacked in signed votes to derive Byzantine-robust connectivity scores from notarization certificates, and deterministically selects a unique weighted leader per epoch from finalized history, thereby improving proposal dissemination reliability under deep fading. Complementing this cognitive adaptation, Wireless Streamlet exploits the single-hop broadcast medium and a deterministic TDMA voting schedule to achieve linear per-epoch on-air transmissions (slot complexity), ensuring deterministic spectral access. To address the communication-storage trade-off, we further propose a coded dual-chain architecture that decouples header-only consensus (State Chain) from payload data (Data Chain). By employing erasure coding and on-chain integrity commitments, the system minimizes redundant spectrum usage for data retrieval while ensuring availability. Experiments show that Wireless Streamlet achieves higher throughput and lower confirmation latency than representative baselines in lossy environments, while substantially reducing per-node storage, demonstrating the efficacy of integrating cognitive sensing into consensus logic.
\end{abstract}


\IEEEpeerreviewmaketitle

\section{Introduction}
\IEEEPARstart{T}{he} Internet of Things (IoT) connects large numbers of resource-constrained devices to sense and act on real-world events \cite{xu2014internet}. In many deployments, especially industrial and multi-stakeholder settings, devices must share data and coordinate actions without relying on a single trusted operator \cite{fernandez2018review}. This makes tamper-evident logging, auditable state updates, and decentralized trust essential requirements rather than optional features \cite{shi2024blockchain_aided_trust}.

Blockchain technology has emerged as a promising solution to these challenges. As a decentralized ledger, it enables secure data sharing among untrusted entities. In the realm of cognitive communications, blockchain has recently been identified as a critical enabler for secure Dynamic Spectrum Sharing (DSS) among multiple operators \cite{li2023multi} and distributed spectrum trading \cite{cheng2023multi}. Furthermore, it serves as a trust substrate for secure task offloading in Mobile Edge Computing (MEC) networks, ensuring reliable resource allocation \cite{du2025secure}. These attributes make blockchain a compelling candidate for securing the ``Edge of Things'' \cite{codesign_2021,zhu2025moe_blockchain}.

Despite its potential, deploying blockchain in wireless edge IoT networks presents unique challenges. First, IoT devices are inherently resource-constrained, making computation-heavy mechanisms such as PoW impractical \cite{shi2023age_of_work}. Second, BFT-style protocols typically incur high all-to-all communication overhead (often $O(n^2)$), which can severely congest bandwidth-limited wireless networks \cite{xu2021wchain}. 
Third, the wireless spectrum is inherently dynamic. Traditional consensus protocols are typically ``spectrum-agnostic'', lacking the cognitive capability to adapt to fading and interference. This challenge is analogous to the problems faced in Distributed Cooperative Spectrum Sensing, where fading disrupts agreement among secondary users \cite{vosoughi2016trust}. Dropped consensus messages in such lossy environments can stall progress, motivating wireless-native designs that explicitly account for unreliable transmissions \cite{zou2023fastconsensus}.

In addition to communication hurdles, storage limitations pose a significant barrier. Maintaining a full replica of the ledger on every IoT node is infeasible. While lightweight client modes exist, they often sacrifice security or decentralization. Consequently, there is a pressing need for a storage-efficient mechanism that maintains data availability and fault tolerance without overwhelming the limited memory of wireless edge devices \cite{zou2024survey}.

To overcome these challenges, this paper proposes \textit{Wireless Streamlet}, a spectrum-aware consensus protocol tailored for wireless edge IoT networks. Building on Streamlet---a streamlined BFT consensus protocol \cite{chan2020streamlet}, we introduce a \textit{Channel-Aware Leader Election (CALE)} mechanism. CALE serves as a verifiable cross-layer cognitive engine that exploits receiver-observed channel quality to improve block propagation reliability under deep fading. Furthermore, to address storage constraints, we introduce an erasure-coding-based dual-chain architecture. This design decouples consensus from storage, allowing nodes to store only fragments of the ledger while ensuring data remains recoverable via Raptor codes \cite{shokrollahi2006raptor}. In summary, the contributions of this work are as follows:

\begin{itemize}
	
    \item \textbf{Wireless-Native Consensus Protocol:} We propose \textit{Wireless Streamlet}, a streamlined consensus protocol optimized for a single-hop wireless broadcast domain in wireless edge IoT. By leveraging the wireless broadcast nature and a deterministic TDMA voting schedule, we reduce the per-epoch on-air transmission cost (slot complexity) to a linear schedule of $(n+1)$ slots, i.e., $\ensuremath{O(n)}$ on-air transmissions per epoch (or $\ensuremath{O(nK_{\mathrm{tx}})}$ with bounded retransmissions), yielding deterministic epoch latency.

	\item \textbf{Cognitive Spectrum-Aware Leader Election (CALE):} We propose CALE, a verifiable channel-aware leader election that serves as a cross-layer cognitive engine: it leverages receiver-measured CSI (signed in votes) to compute Byzantine-robust connectivity scores from notarization certificates, and deterministically selects a unique weighted leader per epoch from finalized history. CALE improves proposal dissemination reliability under deep fading.
	
	\item \textbf{Storage-Efficient Coded Architecture:} We develop a dual-chain architecture coupled with a Raptor-code-based storage scheme. This design decouples consensus (State Chain) from storage (Data Chain), allowing resource-constrained nodes to maintain data availability with minimal storage footprint while ensuring robustness against Byzantine storage failures.
	
	\item \textbf{Rigorous Theoretical Analysis:} We model the consensus layer over packet erasure channels, derive a conservative lower bound on the per-epoch notarization probability, and use a Markov-chain model to obtain the expected time-to-finality under Streamlet's three-consecutive-notarized rule. We also formulate an objective to minimize expected consensus airtime per finalized block by tuning $K_{\mathrm{tx}}$ and CALE, while decoupling payload availability (erasure coding) from consensus liveness.
	
	\item \textbf{Comprehensive Evaluation:} We implement a prototype on the Bamboo benchmarking platform. Experimental results show improved throughput and more stable latency than representative wired-BFT baselines (e.g., PBFT and HotStuff) under high packet loss, validating the effectiveness of our cross-layer design.
\end{itemize}

The remainder of this paper is organized as follows. Section \ref{sec:related_work} reviews related work. Section \ref{sec:system_model} presents the system model. Section \ref{sec:protocol_design} details the Wireless Streamlet protocol and CALE. Section \ref{sec:encoding_scheme} introduces the coded storage architecture. Section \ref{sec:analysis} provides theoretical analysis, followed by evaluation in Section \ref{sec:evaluation} and conclusions in Section \ref{sec:con}.

\section{Related Work}
\label{sec:related_work}

Consensus protocols constitute the backbone of blockchain systems. We review classical BFT protocols, analyses of wireless bottlenecks, and recent wireless-native designs. This discussion highlights the gaps motivating Wireless Streamlet.

\subsection{Classical BFT Protocols}
The Practical Byzantine Fault Tolerance (PBFT) algorithm \cite{castro1999practical} was the first to provide a feasible solution for consensus in asynchronous systems, but it incurs $O(n^2)$ communication (all-to-all copies), making it ill-suited for bandwidth-constrained IoT environments. Recent advances like HotStuff \cite{hotstuff} and Streamlet \cite{chan2020streamlet} streamline the voting structure and can achieve linear communication in implementations with collectors/aggregation, while naive all-to-all dissemination still leads to quadratic point-to-point copies. 


While these protocols perform well in wired Internet settings, they typically rely on abstractions close to reliable point-to-point channels. Blindly applying them to wireless environments can lead to severe throughput degradation due to the lack of cross-layer awareness \cite{iot_challenges_2019}. 

\subsection{Consensus Analysis in Wireless Networks}
Recent literature has analyzed \textit{why} classical consensus degrades in wireless networks, providing theoretical motivation for wireless-aware designs.
Cao et al. \cite{csma_blockchain_2020} modeled MAC-layer effects on blockchain performance and showed that CSMA/CA backoff can amplify latency and increase security risks (e.g., forking) with growing node density. Ni et al. \cite{resource_needed_2022} derived lower bounds on communication resources (e.g., bandwidth and SNR) required to maintain consistency. 
These works identify key bottlenecks such as contention and limited spectral efficiency, but they mainly focus on modeling and bounds rather than proposing a unified protocol that actively mitigates losses. 

\begin{table*}[t!]
	\centering
	\caption{Qualitative Comparison of Consensus/Blockchain Protocols in Wireless IoT Scenarios.
			The reported communication cost is \emph{qualitative} because different works use different metrics (e.g., point-to-point copies vs. on-air slots/time steps).}
	\label{tab:consensus_comparison}
	\small
	\setlength{\tabcolsep}{4pt}
	\renewcommand{\arraystretch}{1.05}
	
	\resizebox{\textwidth}{!}{%
			\begin{tabular}{lcccc}
					\toprule
					\textbf{Protocol} & \textbf{Communication Cost (Metric)}  & \textbf{Leader Election} & \textbf{Channel Awareness} & \textbf{Loss Resilience} \\
					\midrule
					\textbf{PBFT} \cite{castro1999practical}
					& Quadratic (all-to-all message exchange)
					& Round-Robin
					& No
					& Retransmission \\
					\textbf{Streamlet} \cite{chan2020streamlet}
					& Impl.-dependent (often quadratic P2P copies via gossip/echo)
					& Hash-based
					& No
					& Retransmission \\
					\textbf{Wireless-PBFT} \cite{cluster_pbft}
					& Sub-quadratic via clustering (wireless-oriented)
					& Clustering
					& \textbf{Partial}
					& Retransmission \\
					\textbf{Co-Design} \cite{codesign_2021}
					& Varies (PHY/MAC dependent)
					& PHY-based
					& \textbf{Yes (PHY)}
					& Signal processing / PHY-aided \\
					\midrule
					\textbf{wChain} \cite{xu2021wchain}
					& Multihop SINR-oriented; uses backbone/aggregation (slot-based analysis)
					& Leader-based (BFT-style)
					& \textbf{Yes (SINR/MAC-aware)}
					& Aggregation + recovery mechanisms \\
					\textbf{Fast Consensus} \cite{zou2023fastconsensus}
					& Near-optimal time-steps under SINR (randomized distributed consensus)
					& Leaderless / randomized
					& \textbf{Yes (SINR model)}
					& Built-in tolerance to unreliable transmissions \\
					\textbf{BLOWN} \cite{blown2023}
					& Single-hop; designed under adversarial SINR (security-driven overhead)
					& PoC-based (channel-driven)
					& \textbf{Yes (adversarial SINR)}
					& Anti-jamming / adversarial robustness \\
					\midrule
					\textbf{Wireless Streamlet} (ours)
					& \textbf{Linear on-air slots/epoch (TDMA, header-only consensus)}
					& \textbf{CALE (QC-derived CSI-weighted, unique leader)}
					& \textbf{Yes (cross-layer)}
					& \textbf{Bounded retransmission (consensus) + erasure-coded data plane} \\
					\bottomrule
				\end{tabular}%
		}
\end{table*}

\subsection{Optimization for Wireless Consensus}
To overcome these limitations, research explores topology control and cross-layer co-design.

\subsubsection{Topology and Deployment Optimization}
Sun et al. \cite{deployment_2019} analyzed node deployment strategies to improve coverage and connectivity. Hierarchical protocols such as G-PBFT/Wireless-PBFT \cite{cluster_pbft} partition the network into clusters to reduce overhead. However, deployment optimization is often static, and clustering methods typically treat the wireless link as a black box and rely on ARQ-style retransmissions. 

\subsubsection{Communication-Consensus Co-Design}
Seo et al. \cite{codesign_2021} proposed a co-design framework leveraging physical-layer properties (e.g., over-the-air computation) to accelerate consensus. While promising, PHY-layer co-design often requires strict synchronization or waveform-level modifications, limiting compatibility with commodity IoT hardware. 

\subsubsection{Coded Blockchain and Sharding}
Erasure coding has been explored to address storage and scalability. Polyshard \cite{li2020polyshard}, coded Merkle trees \cite{yu2020coded}, and coded blockchain designs \cite{yang2022coded} improve data availability and reduce storage costs, primarily in wired/sharded settings. Our work adapts coded storage to lossy wireless links and couples it with wireless-native consensus. 

\subsubsection{Wireless-Native Blockchain and Consensus Protocols}
Several protocols are designed explicitly for wireless settings. wChain targets multihop networks under the SINR model and improves efficiency via communication backbones and aggregation primitives~\cite{xu2021wchain}. Fast Consensus studies distributed agreement under unreliable SINR transmissions and achieves near-optimal time complexity in terms of wireless time steps~\cite{zou2023fastconsensus}. BLOWN focuses on a single-hop permissioned blockchain under an adversarial SINR model and emphasizes robustness against jamming and related attacks~\cite{blown2023}.

In contrast, Wireless Streamlet targets a single-hop broadcast IoT cluster and tailors the standard Streamlet BFT protocol to the wireless edge. Specifically, it replaces point-to-point vote dissemination with a deterministic TDMA voting schedule that explicitly \emph{exploits the wireless one-to-many broadcast nature} (one on-air transmission can be received by many nodes), thereby making epoch progress depend on \emph{on-air slot usage} rather than all-to-all message copies. Moreover, Wireless Streamlet introduces channel-aware leader election to improve block propagation under deep fading, and couples consensus with erasure-coded data availability via a coded dual-chain architecture to meet IoT storage constraints under lossy (non-adversarial) packet erasures.

\subsection{Qualitative Comparison}
Wireless Streamlet adopts a lightweight cross-layer design compatible with standard transceivers. Table \ref{tab:consensus_comparison} summarizes key qualitative differences. Since these works optimize different objectives and use different cost metrics (e.g., P2P copies vs. on-air slots vs. time steps), we focus on qualitative comparison and evaluate against representative wired-BFT baselines under our target setting. 

\section{System Model}
\label{sec:system_model}

Before detailing the proposed consensus protocol, we establish the node/adversary assumptions, the wireless channel model, and the timing model. Unlike traditional blockchain systems that assume reliable wired connectivity, our model explicitly captures wireless packet erasures and their impact on consensus availability and latency. We focus on a wireless edge IoT setting where the consensus plane forms a single-hop omni-directional broadcast domain; thus, a single on-air transmission can be overheard by all nodes subject to packet erasures.

\subsubsection{Node Model}
We consider a network of $n$ wireless IoT nodes denoted by the set~$\mathcal{N} = \{1,2,\dots,n\}$. The system tolerates Byzantine faults: up to $f$ nodes may behave maliciously or fail arbitrarily, satisfying $n \ge 3f+1$, which is standard for Byzantine fault tolerant consensus under (partial/weak) synchrony \cite{dwork1988consensus,castro1999practical}.
Each node $i \in \mathcal{N}$ possesses a unique public-private key pair $(\mathsf{pk}_i,\mathsf{sk}_i)$ for digital signatures, ensuring authenticated and non-repudiable messages \cite{castro1999practical}.

\subsubsection{Wireless Channel Model (Packet Erasure Channel)}
A critical distinction of our system is the explicit modeling of the communication medium. We model wireless communication as a packet erasure channel (PEC): in a given slot, a packet is either successfully decoded or erased (lost) \cite{cover2006elements, fragouli2007network}.
We allow heterogeneous link qualities: a transmission attempt from node $i$ to node $j$ succeeds with probability $p_{i\to j}\in(0,1]$.

\begin{itemize}
	\item \textbf{Packet Loss/PER:} The per-attempt success probability can be expressed as
	\begin{equation}
		p_{i\to j} = 1-\mathrm{PER}_{i\to j},
	\end{equation}
	where $\mathrm{PER}_{i\to j}$ is determined by physical-layer conditions (such as SNR under Rayleigh or Rician fading) and MAC-layer interference; we abstract these effects into the PEC parameter $p_{i\to j}$ \cite{goldsmith2005wireless}.
	\item \textbf{Non-degenerate honest connectivity:} For liveness analysis, we assume there exists a constant $p_H>0$ such that for any epoch (and any history), any honest sender-to-honest receiver transmission attempt succeeds with probability at least $p_H$. This type of Bernoulli packet-loss abstraction with a strictly positive success probability is standard in the analysis of networked systems with intermittent/lossy communication \cite{charron2009heard}.
	
    \item \textbf{Consensus-plane broadcast domain:} We assume that \emph{consensus-plane} transmissions for proposal/vote packets occur within a single-hop omni-directional wireless broadcast domain, i.e., every pair of consensus nodes is within mutual radio range so that a transmission can be overheard by all other nodes subject to packet erasures. Multi-hop forwarding, inter-cluster communication, and gateway/backhaul connectivity (used by the storage plane, if present) are not modeled and are out of scope for the consensus analysis.

\end{itemize}

\begin{figure*}[t]
	\centering
	\resizebox{2\columnwidth}{!}{%
		\begin{tikzpicture}[
			node distance=1.5cm and 3cm,
			base_node/.style={
				rectangle,
				draw=gray!40,
				thick,
				rounded corners=3pt,
				minimum width=2.2cm,
				minimum height=0.8cm,
				align=center,
				font=\small\bfseries,
				drop shadow={opacity=0.15, shadow xshift=1pt, shadow yshift=-1pt}
			},
			leader_node/.style={
				base_node,
				top color=white,
				bottom color=ieee_orange!15,
				draw=ieee_orange!80
			},
			worker_node/.style={
				base_node,
				top color=white,
				bottom color=ieee_blue!10,
				draw=ieee_blue!60
			},
			arrow_style/.style={
				-{Stealth[length=2mm]},
				draw=gray!60,
				thick
			},
			stage_label/.style={
				font=\bfseries\small,
				text=black!80,
				yshift=-0.5cm
			}
			]
			
			\node[leader_node] (L1) at (0,0) {Leader};
			\node[worker_node] (N1_1) at (0,-1.2) {Node 1};
			\node[worker_node] (N2_1) at (0,-2.4) {Node 2};
			\node[worker_node] (N3_1) at (0,-3.6) {Node 3};
			
			\node[leader_node] (L2) at (5,0) {Leader};
			\node[worker_node] (N1_2) at (5,-1.2) {Node 1};
			\node[worker_node] (N2_2) at (5,-2.4) {Node 2};
			\node[worker_node] (N3_2) at (5,-3.6) {Node 3};
			
			\node[leader_node] (L3) at (10,0) {Leader};
			\node[worker_node] (N1_3) at (10,-1.2) {Node 1};
			\node[worker_node] (N2_3) at (10,-2.4) {Node 2};
			\node[worker_node] (N3_3) at (10,-3.6) {Node 3};
			
			\node[right=0.5cm of L3, font=\footnotesize\itshape] {Notarized};
			\node[right=0.5cm of N1_3, font=\footnotesize\itshape] {Notarized};
			\node[right=0.5cm of N2_3, font=\footnotesize\itshape] {Notarized};
			\node[right=0.5cm of N3_3, font=\footnotesize\itshape] {Notarized};
			
			\draw[arrow_style, ieee_orange] (L1.east) -- (L2.west);
			\draw[arrow_style, ieee_orange] (L1.east) -- (N1_2.west);
			\draw[arrow_style, ieee_orange] (L1.east) -- (N2_2.west);
			\draw[arrow_style, ieee_orange] (L1.east) -- (N3_2.west);
			
			\draw[arrow_style, ieee_blue] (N1_2.east) -- (L3.west);
			\draw[arrow_style, ieee_blue] (N2_2.east) -- (N1_3.west);
			\draw[arrow_style, ieee_blue] (N3_2.east) -- (N2_3.west);
			\draw[arrow_style, ieee_blue] (L2.east) -- (N3_3.west);
			
			\draw[arrow_style, ieee_blue!30] (N1_2.east) -- (N2_3.west);
			\draw[arrow_style, ieee_blue!30] (N2_2.east) -- (L3.west);
			
			\draw[dotted, thick, gray] (2.5, 0.5) -- (2.5, -4.5);
			\draw[dotted, thick, gray] (7.5, 0.5) -- (7.5, -4.5);
			
			\node[stage_label] at (0, -4.2) {1. Block Proposal};
			\node[stage_label] at (5, -4.2) {2. Node Voting};
			\node[stage_label] at (10, -4.2) {3. Notarization/Finality Update};
			
			\draw[->, thick, gray] (-1, 0.8) -- (10, 0.8) node[right] {Logical Flow (Epoch)};
			
		\end{tikzpicture}%
	} 
	\caption{Logical execution flow of Wireless Streamlet within one epoch. The consensus logic follows Streamlet  \cite{chan2020streamlet}, while the physical transmission is adapted to wireless broadcast and a TDMA schedule.}
	\label{fig:streamlet_stages}
\end{figure*}
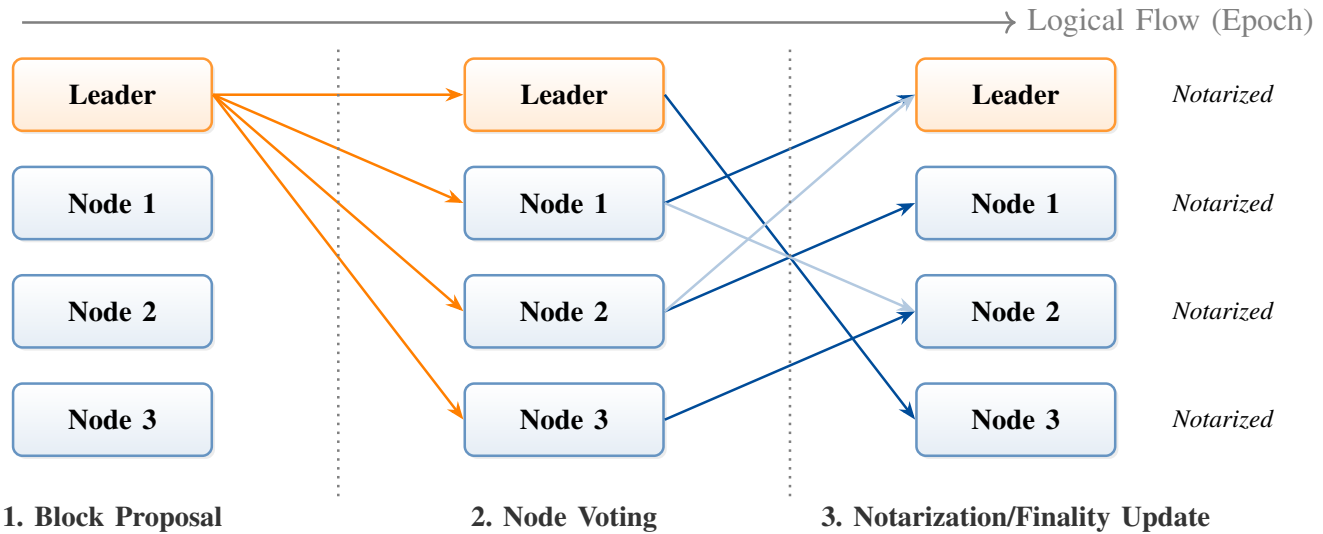

\subsubsection{Broadcast Reception Model}
When a leader broadcasts a packet, different honest receivers may observe different outcomes due to heterogeneous links. Let $X$ denote the number of honest nodes that successfully receive a leader's broadcast in a given slot. With heterogeneous per-receiver success probabilities, $X$ is the sum of independent (but not identically distributed) Bernoulli random variables, i.e., it follows a Poisson-binomial distribution \cite{johnson2005univariate}. For tractable closed-form bounds used later, if the leader-to-honest links satisfy $p_{L\to j}\ge p_H$ for all honest receivers $j$, then
\begin{equation}
	\Pr[X\ge k] \ge \sum_{x=k}^{h} \binom{h}{x} (p_H)^x(1-p_H)^{h-x},
\end{equation}
where $h=n-f$ is the number of honest nodes. Reception events across receivers may be correlated in practice (e.g., common interference). Our analysis uses the above independent lower bound to obtain conservative, analytically tractable guarantees.

\subsubsection{Timing Model (Weak Synchrony with Epochs)}
We adopt a weakly synchronous timing model. Time is divided into epochs, each consisting of a fixed sequence of slots. If a message is not erased by the wireless channel, it is delivered within a bounded delay; otherwise it is treated as non-delivery for that epoch. This abstraction is consistent with classical partial/weak synchrony models used for liveness in Byzantine consensus \cite{dwork1988consensus} and aligns with round/epoch-based blockchain-style protocols \cite{chan2020streamlet}.

\subsubsection{Complexity Metrics}
To avoid ambiguity, we distinguish four distinct metrics: first, logical message objects per epoch, comprising proposals and votes; second, point-to-point message copies, applicable if broadcast is realized by unicast replication; third, wireless on-air transmissions---also referred to as slot complexity---in our single-hop broadcast domain where a single transmission may be overheard by many nodes; and fourth, bit complexity, defined as the total transmitted bytes. Unless explicitly stated otherwise, any $O(\cdot)$ complexity in this paper refers to on-air transmissions, measured in TDMA slots per epoch within the consensus plane.

\section{Wireless Streamlet Protocol Design}
\label{sec:protocol_design}

This section presents the Wireless Streamlet protocol. We first summarize the Streamlet-style consensus logic, and then introduce two wireless-native transmission adaptations: (i) replacing application-layer echoing with broadcast-based vote dissemination under a single-hop broadcast-domain assumption, and (ii) a TDMA-based voting schedule to avoid collisions and provide deterministic epoch latency. We then present CALE, a verifiable cross-layer leader election mechanism that derives connectivity weights from receiver-measured CSI piggybacked in signed votes.

\subsection{Protocol Overview and Logic}
Wireless Streamlet runs with $n$ nodes and tolerates up to $f$ Byzantine nodes, assuming the standard threshold $n \ge 3f+1$ \cite{dwork1988consensus,castro1999practical,chan2020streamlet}. Time is divided into epochs, and each epoch is further divided into a fixed sequence of time slots (defined in Section~\ref{subsec:tdma_voting}). This slotting matches the weak-synchrony model introduced in Section~\ref{sec:system_model}.

The protocol is a state-machine replication style blockchain consensus with three block states: \textit{proposed}, \textit{notarized}, and \textit{finalized}, following Streamlet \cite{chan2020streamlet}. Nodes follow the \emph{longest notarized chain} rule. As illustrated in Fig.~\ref{fig:streamlet_stages}, each epoch consists of three functional stages:

\begin{itemize}
	\item \textbf{Proposal:} A leader proposes a block extending the longest notarized chain.
    \item \textbf{Voting:} Upon receiving a valid proposal from the designated leader, an honest node verifies the block, including ancestry and signature, and broadcasts a signed vote (piggybacking a small receiver-measured CSI tag for CALE).
	\item \textbf{Confirmation:} A block is \textit{notarized} once it gathers at least $2f+1$ votes. A block becomes \textit{finalized} when it is the third block in a sequence of three consecutive notarized blocks in consecutive epochs, as in Streamlet \cite{chan2020streamlet}.
\end{itemize}

\subsection{Wireless-Native Transmission Strategy}
While the logical state machine mirrors Streamlet, the transmission mechanism must be redesigned for wireless edge IoT networks. We introduce two adaptations: reducing redundant vote relaying by leveraging broadcast reception, and enforcing a TDMA schedule for collision-free voting.

\subsubsection{Leveraging Broadcast to Replace Echoing}
\label{subsec:broadcast_no_echo}

In the original Streamlet design intended for point-to-point wired networks, nodes may re-broadcast (echo) votes to ensure that votes quickly become visible to many parties even under selective delivery or scheduling effects \cite{chan2020streamlet}.

In our wireless setting, we assume a single-hop broadcast domain where nodes transmit using omni-directional antennas over a shared channel. Under this assumption, a transmitted packet is physically observable by all nodes within range, subject to independent packet erasures as modeled in Section~\ref{sec:system_model}. This broadcast property provides a \textit{natural dissemination mechanism} for vote signatures: each voter transmits its vote once, and all nodes can overhear it (if not erased).

\begin{itemize}
	\item \textbf{Threat model by wireless nature:}	Wireless broadcast does \emph{not} make selective delivery impossible in an absolute sense (e.g., via directional antennas or power control). Therefore, our design assumes an IoT cluster without beamforming-capable adversaries and focuses on the common omni-directional setting. In this regime, broadcast overhearing substantially reduces the need for explicit application-layer echoing, and the TDMA schedule further increases the probability that a transmitted vote is heard by many nodes within the same epoch.
	
	\item \textbf{Pruning Streamlet for wireless broadcast:}
	As illustrated in Fig.~\ref{fig:selective_transmission}, selective transmission to only a subset of nodes is easy in wired point-to-point networks, whereas wireless broadcasts are naturally overheard by many receivers, subject to erasures. We therefore remove the explicit echo phase of the original Streamlet and let each node broadcast its vote once in its assigned TDMA slot. This change reduces the wireless on-air cost per epoch, i.e., the slot complexity, from all-to-all amplification to linear in $n$, conserving airtime while preserving the notarization logic.
	
\end{itemize}

\begin{figure}[t]
	\centering
    \resizebox{1\columnwidth}{!}{
	\begin{tikzpicture}[
		node distance=1.2cm and 2.5cm,
		honest_node/.style={
			rectangle,
			draw=ieee_blue!80,
			fill=ieee_blue!5,
			thick,
			rounded corners=2pt,
			minimum width=1.8cm,
			minimum height=0.7cm,
			font=\small\bfseries
		},
		malicious_node/.style={
			rectangle,
			draw=red!80,
			fill=red!5,
			thick,
			rounded corners=2pt,
			minimum width=1.8cm,
			minimum height=0.7cm,
			font=\small\bfseries,
			text=red!80
		},
		arrow_ok/.style={
			-{Stealth},
			draw=ieee_blue,
			thick
		},
		arrow_bad/.style={
			-{Stealth},
			draw=red,
			dashed,
			thick
		},
		cross/.style={
			cross out,
			draw=red,
			minimum size=4mm,
			thick
		}
		]
		
		\node[honest_node] (H1) at (0, 0) {Node A};
		\node[honest_node] (H2) at (0, -1.2) {Node B};
		\node[malicious_node] (M1) at (0, -2.8) {Malicious X};
		
		\node[honest_node] (R1) at (4, 0) {Node A};
		\node[honest_node] (R2) at (4, -1.2) {Node B};
		\node[honest_node] (R3) at (4, -2.4) {Node C};
		\node[honest_node] (R4) at (4, -3.6) {Node D};
		
		\draw[arrow_bad] (M1.east) -- (R1.west) node[midway, above, sloped, font=\tiny, text=red] {Msg 1};
		\draw[arrow_bad] (M1.east) -- (R2.west);
		
		\draw[arrow_bad, opacity=0.3] (M1.east) -- (R3.west) node[midway, fill=white] (X1) {};
		\node[cross] at (X1) {};
		
		\draw[arrow_bad, opacity=0.3] (M1.east) -- (R4.west) node[midway, fill=white] (X2) {};
		\node[cross] at (X2) {};
		
		\node[align=center, font=\footnotesize] at (0, -4.2) {\textbf{Sender Phase}};
		\node[align=center, font=\footnotesize] at (4, -4.2) {\textbf{Receiver Phase}};
		
		\draw[dotted, thick, gray] (2, 0.5) -- (2, -4.5);
		
	\end{tikzpicture}
    }

	\caption{Selective transmission (delivery to only a subset) is easy in wired point-to-point networks. In a single-hop omni-directional wireless broadcast domain, transmissions are naturally overheard by many nodes (subject to erasures), reducing the need for an explicit application-layer echo phase.}

	\label{fig:selective_transmission}
\end{figure}
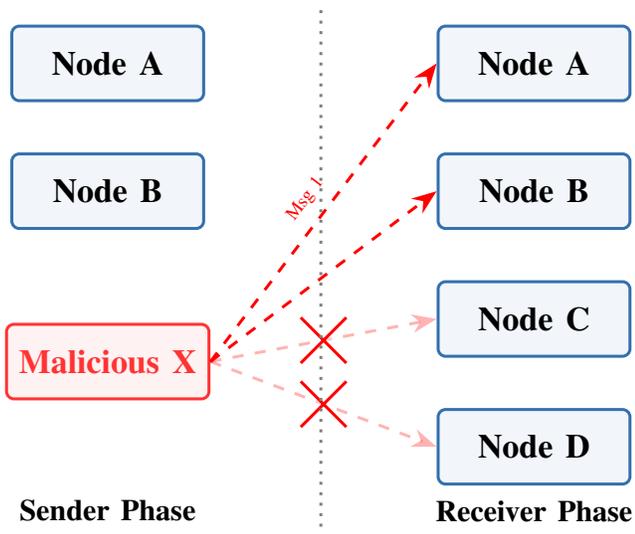

\subsubsection{TDMA-Based Deterministic Voting}
\label{subsec:tdma_voting}

A second challenge arises in the voting stage. If all nodes broadcast votes using contention-based random access (e.g., CSMA/CA), packet collisions and random backoff can cause highly variable latency, which is undesirable for epoch-based consensus \cite{ieee80211_2020}.

To mitigate this, we adopt a TDMA schedule. Fig.~\ref{fig:tdma_schedule} illustrates the slot structure. Each epoch is divided into $(n+1)$ slots: the leader uses Slot~0 to broadcast the proposal, followed by one dedicated voting slot for each node. The epoch duration is therefore deterministic:
\begin{equation}
	T_{\mathrm{epoch}} = (n+1)T_{\mathrm{slot}} + T_{\mathrm{guard}}.
	\label{eq:epoch_duration}
\end{equation}
which implies a linear slot complexity of $\ensuremath{O(n)}$ on-air transmissions per epoch in the consensus plane. This design provides:
\begin{itemize}
	\item \textbf{Collision avoidance:} orthogonal time slots eliminate vote collisions.
	\item \textbf{Deterministic latency:} the maximum time to complete one epoch is bounded by $T_{\mathrm{epoch}}$.
\end{itemize}

\begin{figure}[t]
	\centering
    \resizebox{1\columnwidth}{!}{
	\begin{tikzpicture}[xscale=1.3, yscale=0.8, font=\small]
		\draw[->, thick, gray] (0, 0) -- (5, 0) node[right] {Time $t$};
		\draw[thick, gray] (0, 0) -- (0, 5);
		
		\node[anchor=east] at (0, 4) {Leader};
		\node[anchor=east] at (0, 3) {Node 1};
		\node[anchor=east] at (0, 2) {Node 2};
		\node[anchor=east] at (0, 1) {Node 3};
		
		\tikzstyle{slot_leader} = [fill=ieee_orange!30, draw=ieee_orange, thick, rounded corners=2pt]
		\tikzstyle{slot_node} = [fill=ieee_blue!20, draw=ieee_blue, thick, rounded corners=2pt]
		
		\filldraw[slot_leader] (0.2, 3.6) rectangle (1.8, 4.4); \node at (1, 4) {Proposal};
		\filldraw[slot_node] (1.8, 2.6) rectangle (2.8, 3.4); \node at (2.3, 3) {Vote};
		\filldraw[slot_node] (2.8, 1.6) rectangle (3.8, 2.4); \node at (3.3, 2) {Vote};
		\filldraw[slot_node] (3.8, 0.6) rectangle (4.8, 1.4); \node at (4.3, 1) {Vote};
		
		\draw[dashed, gray!50] (1.8, 0) -- (1.8, 4.5);
		\draw[dashed, gray!50] (2.8, 0) -- (2.8, 4.5);
		\draw[dashed, gray!50] (3.8, 0) -- (3.8, 4.5);
		\draw[dashed, gray!50] (4.8, 0) -- (4.8, 4.5);
		
		\node[below, gray] at (1, 0) {Slot 0};
		\node[below, gray] at (2.3, 0) {Slot 1};
		\node[below, gray] at (3.3, 0) {Slot 2};
		\node[below, gray] at (4.3, 0) {Slot 3};
	\end{tikzpicture}
    }

	\caption{TDMA-based sequential transmission schedule (illustrative). Slot~0 is reserved for the leader's proposal; each node broadcasts its vote in a dedicated slot.}

\label{fig:tdma_schedule}
\end{figure}
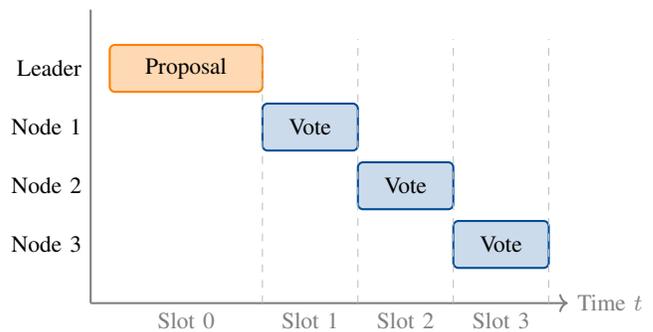

\subsection{Channel-Aware Leader Election (CALE)}
\label{subsec:cale}

In canonical Streamlet, leaders are selected pseudo-randomly per epoch \cite{chan2020streamlet}, which can be inefficient in dynamic wireless environments where link quality fluctuates over time. A ``spectrum-ignorant'' protocol may elect a leader in deep fading, wasting the proposal slot and reducing the probability of collecting a quorum within the same epoch.

We propose CALE, a lightweight and cognitive cross-layer mechanism that biases leader selection toward nodes with better broadcast connectivity. Crucially, CALE is designed to be \emph{verifiable}: it derives leader connectivity from receiver-measured channel state information (CSI) that is cryptographically bound to signed votes, so that other nodes can validate the score from the notarization certificate.

\subsubsection{Receiver-Side Spectrum Sensing and Verifiable Quality Metric}
During Slot~0 of epoch $e$, the designated leader $L(e)$ broadcasts the proposal (a compact State Block header). Each node $j$ that successfully receives the proposal can measure a receiver-side CSI indicator, such as SNR or RSSI, from standard PHY measurements (e.g., preamble/pilots). Let
$\tilde{\gamma}_{L(e)\to j}(e)$ denote node $j$'s measured CSI for the proposal transmission from $L(e)$ in epoch $e$.\footnote{In practice, $\tilde{\gamma}$ can be a low-cost quantized SNR/RSSI (e.g., 8--16 bits), or even an MCS index/link-quality indicator (LQI) provided by commodity radios. This adds negligible overhead to the vote message and requires no additional transmissions.}  

\textbf{Vote Format with CSI Tag}:
When node $j$ votes in its TDMA slot for the proposal of epoch $e$, it broadcasts a signed vote that \emph{piggybacks} the measured CSI:
\begin{equation}
	\mathsf{Vote}_j(e) \triangleq \Big( e,\; \mathsf{Hash}(B_e),\; \tilde{\gamma}_{L(e)\to j}(e) \Big)_{\mathsf{sk}_j},
	\label{eq:vote_with_csi}
\end{equation}
where $B_e$ is the proposed State Block header in epoch $e$ and $(\cdot)_{\mathsf{sk}_j}$ denotes a digital signature by node $j$. Since $\tilde{\gamma}_{L(e)\to j}(e)$ is measured at the receiver and bound to $j$'s signature, a Byzantine leader cannot fabricate or inflate these CSI tags.

\textbf{Verifiable Connectivity Score from QC}:
Let $\mathsf{QC}_e$ denote the notarization certificate (quorum) for $B_e$, and let $\mathcal{Q}_e$ be the set of voters whose signed votes are included in $\mathsf{QC}_e$ (thus $|\mathcal{Q}_e|\ge 2f+1$ when notarized). We define the leader's spectral connectivity score as a \emph{Byzantine-robust} aggregate over the CSI tags contained in $\mathsf{QC}_e$:
\begin{equation}
	\Omega_{L(e)}(e)
	\triangleq
	\operatorname{median}_{j \in \mathcal{Q}_e}\;
	\log_2\!\Big( 1 + \tilde{\gamma}_{L(e)\to j}(e) \Big).
	\label{eq:connectivity_score}
\end{equation}
The median aggregator tolerates up to $f$ arbitrarily misreported CSI values among $\mathcal{Q}_e$ (Byzantine voters), while remaining fully verifiable from $\mathsf{QC}_e$. Therefore, CALE eliminates the ``self-report'' vulnerability: $\Omega_{L(e)}(e)$ can be recomputed and checked by any node that verifies $\mathsf{QC}_e$.

\textbf{Maintaining Per-node Weights from Finalized History}:
Because $\Omega_i$ is observed only when node $i$ serves as leader, each node maintains a deterministic weight table from the finalized State Chain. Let $\Omega_i^{\mathrm{fin}}(e)$ denote node $i$'s most recent \emph{finalized} connectivity score prior to epoch $e$ (initialized to a common constant, e.g., $\Omega_i^{\mathrm{fin}}(0)=1$ if unseen). Define a global baseline
\begin{equation}
	\bar{\Omega}(e) \triangleq \frac{1}{n}\sum_{i=1}^{n} \Omega_i^{\mathrm{fin}}(e),
	\quad
	w_i(e) \triangleq \frac{\max\{\Omega_i^{\mathrm{fin}}(e),\Omega_{\min}\}}{\bar{\Omega}(e)},
	\label{eq:weight_def}
\end{equation}
where $\Omega_{\min}>0$ is a small floor to avoid degenerate weights at cold start. Since $\Omega_i^{\mathrm{fin}}(e)$ is derived from \emph{finalized} certificates, all honest nodes compute identical $w_i(e)$.

\subsubsection{Deterministic Weighted Unique-Leader Rule}
\label{subsec:cale_unique_leader}

The threshold-based ``self-election'' rule (which may yield zero or multiple leaders) is incompatible with our TDMA design where Slot~0 must have a unique transmitter. CALE therefore uses a deterministic \emph{unique-winner} rule that all nodes can compute consistently.

For epoch $e$, each node $i$ derives a public pseudo-random value
$u_i(e) \in (0,1]$ using a normalized hash:
$u_i(e) \triangleq \mathsf{Hash}_{nor}(e \,\|\, \mathsf{pk}_i)$.
We then define a weighted priority score
\begin{equation}
\begin{aligned}
	u_i(e) &\triangleq \mathsf{Hash}_{nor}(e \,\|\, \mathsf{pk}_i),\\
	\rho_i(e) &\triangleq \frac{-\ln u_i(e)}{w_i(e)^{\alpha}}, \qquad \alpha>0,\\
	L(e) &\triangleq \arg\min_{i \in \mathcal{N}} \rho_i(e),
\end{aligned}
\label{eq:cale_rule}
\end{equation}
where $\alpha>0$ is a sensitivity parameter. A deterministic tie-breaker (e.g., by smaller $\mathsf{pk}_i$) can be applied if needed. This rule guarantees \emph{exactly one} leader per epoch by construction, while biasing selection toward nodes with larger (verifiable) weights $w_i(e)$. Intuitively, larger $w_i(e)$ decreases $\rho_i(e)$ in expectation, increasing node $i$'s chance of being selected.

\begin{remark}[Consistency under temporary notarized forks]
To avoid inconsistent leader derivation caused by transient notarized forks, all nodes derive the weight table $\{w_i(e)\}$ from the latest \emph{finalized} State Chain checkpoint available at the beginning of epoch $e$, and hold it constant between checkpoints. Since finalized history is common to all honest nodes under Streamlet safety \cite{chan2020streamlet}, this ensures that all honest nodes compute the same unique leader $L(e)$ and thus preserve collision-free TDMA Slot~0.
\end{remark}

\begin{remark}[Overhead of CALE]
CALE introduces no additional on-air transmissions: the only extra information is the small CSI tag $\tilde{\gamma}$ piggybacked in the existing vote packet (default 1~KB), and the leader selection computation is local.
\end{remark}

\begin{table*}[t]
	\centering
	\caption{Communication overhead comparison in a wireless shared medium (decentralized setting). ``Transmissions'' refer to \emph{on-air transmission attempts} in the single-hop wireless medium.}
	\label{tab:complexity_comparison}
	\begin{tabular}{lccc}
		\toprule
		\textbf{Protocol} &
		\textbf{Per-epoch transmissions} &
		\textbf{Channel access} &
		\textbf{Collision behavior} \\
		\midrule
		Standard Streamlet + gossip/echo (wireless, CSMA/CA)  &
	    \ensuremath{O(n^2)} (worst-case) &
		Random access (e.g., CSMA/CA)&
		High contention / backoff \cite{bianchi2000performance} \\
		\textbf{Wireless Streamlet} &
		\ensuremath{O(n)} (or \ensuremath{O(n\,K_{\mathrm{tx}})}) &
		TDMA schedule &
		\textbf{Collision-free by design} \\
		\bottomrule
	\end{tabular}
\end{table*}

\subsection{Communication Efficiency Analysis}
\label{subsec:complexity}

This subsection compares the communication overhead of Wireless Streamlet against a baseline deployment of Streamlet in a decentralized wireless ad-hoc setting; Table~\ref{tab:complexity_comparison} summarizes the key differences in per-epoch transmissions and wireless-medium behavior. We distinguish (i) \emph{logical message complexity} (how many protocol messages are generated) from (ii) \emph{wireless medium behavior} (contention/collisions and their impact on effective latency).

\subsubsection{Baseline: Standard Streamlet over Decentralized Wireless Ad-Hoc}
Streamlet's notarization requires that votes become visible to a large fraction of nodes (a quorum of at least \ensuremath{2f+1}) in each epoch \cite{chan2020streamlet}. In a decentralized wireless ad-hoc network without a trusted aggregator, a natural implementation is that each node disseminates its vote and, in many designs, re-disseminates votes it learns (echo/gossip) to accelerate visibility under selective delivery and scheduling uncertainty \cite{chan2020streamlet}.

\begin{itemize}
	\item \textbf{Logical objects vs. dissemination cost:} Streamlet generates one proposal and (at most) one vote per node per epoch, i.e., \emph{logical message objects} are $\ensuremath{O(n)}$. However, making votes visible to many parties in a decentralized setting typically relies on all-to-all dissemination, often realized via unicast replication or gossip protocols. Consequently, this approach can incur $O(n^2)$ point-to-point copies or on-air transmissions per epoch in the worst case.

	\item \textbf{Wireless MAC penalty under random access:} In a shared wireless medium using contention-based access (e.g., CSMA/CA as in IEEE~802.11 DCF), concurrent transmission attempts lead to collisions and backoff, and the throughput degrades as the number of contending nodes increases \cite{ieee80211_2020,bianchi2000performance}. Consequently, the \emph{effective} epoch latency can grow sharply with \ensuremath{n}, beyond what is suggested by message counts alone.
\end{itemize}

\subsubsection{Proposed Wireless Streamlet}
Wireless Streamlet reduces both the number of transmissions and the contention-induced overhead by aligning consensus communication with wireless broadcast and a collision-free schedule.

\begin{itemize}
	\item \textbf{No explicit echo phase (broadcast-based vote visibility):} Under the single-hop omni-directional broadcast-domain assumption (Section~\ref{subsec:broadcast_no_echo}), each vote is transmitted once and can be overheard by all nodes within range unless channel erasures occur. This removes the need for an explicit application-layer echoing step that otherwise amplifies vote traffic.
	\item \textbf{Deterministic slot complexity via TDMA:} Each epoch is divided into exactly \ensuremath{(n+1)} transmission opportunities (one proposal slot plus \ensuremath{n} vote slots). Hence, the number of scheduled on-air transmissions per epoch scales linearly with \ensuremath{n}. TDMA is a standard approach to eliminate collisions by orthogonalizing channel access in time \cite{tse2005fundamentals}.
	\item \textbf{Retransmission factor \ensuremath{K_{\mathrm{tx}}}:} To cope with packet erasures, each scheduled slot may allow up to \ensuremath{K_{\mathrm{tx}}} retransmissions. Importantly, \ensuremath{K_{\mathrm{tx}}} is chosen to meet a reliability target under the erasure model (Section~\ref{sec:system_model}) and does not introduce an all-to-all amplification; thus the total scheduled transmissions per epoch are \ensuremath{O(n\,K_{\mathrm{tx}})}.
\end{itemize}
Note that the number of logical votes remains $\ensuremath{O(n)}$; the gain comes from realizing dissemination with one-to-many wireless broadcast and collision-free scheduling.

Overall, compared with a decentralized Streamlet deployment that induces \ensuremath{O(n^2)} transmissions and suffers contention under CSMA/CA, Wireless Streamlet achieves \ensuremath{O(n)} scheduled transmissions (or \ensuremath{O(n\,K_{\mathrm{tx}})} with bounded retransmissions) and avoids collision-induced backoff by construction.

\section{Lightweight Block Encoding and Storage Architecture}
\label{sec:encoding_scheme}

To address the storage and bandwidth constraints of wireless IoT networks, we propose a lightweight data architecture for Wireless Streamlet that \emph{decouples consensus from bulk data storage}. The key idea is to keep the consensus-facing ledger (State Chain) compact, while storing transaction payloads (Data Chain) in an erasure-coded form across storage nodes. Crucially, the State Chain \emph{cryptographically commits} to the corresponding Data Chain content, so that later decoding and auditing remain verifiable.

\subsection{Why Fountain/Raptor Codes for IoT Payloads}
\label{subsec:raptor_rationale}

Raptor codes are adopted as a practical rateless erasure code with near-linear-time encoding/decoding, avoiding the heavy finite-field operations typical of fixed-rate MDS codes (e.g., Reed--Solomon) on low-power IoT devices \cite{wicker1999reed,shokrollahi2006raptor}. Their rateless property lets the storage layer generate and distribute just enough symbols to meet an availability target, while the consensus layer stores only compact commitments to the payload (Section~\ref{subsec:dual_chain_commitment}). In practice, this reduces computation on edge devices and lowers per-node payload storage by distributing coded symbols across storage participants.

\subsection{Dual-Chain Architecture and Node Roles}
\label{subsec:dual_chain_commitment}

To prevent large payload dissemination from congesting the wireless broadcast domain, we bifurcate the ledger into two logically linked chains (Fig.~\ref{fig:block_structure}):

\begin{enumerate}
	\item \textbf{Data Chain (payload storage):} contains the transaction payload for each block, encoded into symbols and stored across storage nodes.
	\item \textbf{State Chain (consensus metadata):} contains only lightweight metadata required for consensus and verification, including:
	(i) the parent pointer, (ii) a state commitment (e.g., a Merkle state root), and
	(iii) a \emph{data commitment} that binds the State Block to the corresponding payload encoding.
\end{enumerate}

Each payload block $p$ is anchored to the State Chain by including in its corresponding State Block (i) a payload identifier $id_p$ (e.g., the hash of a canonical payload header) and (ii) a commitment $C_p$ to the verification structure of its encoded symbols (Section~\ref{subsec:integrity_hashing}). Since $(id_p, C_p)$ is notarized/finalized by Streamlet, it serves as an immutable consensus anchor: any subsequent payload retrieval or reconstruction must verify against this pair.

The separation yields a natural split of responsibilities, as illustrated in Fig.~\ref{fig:system_architecture}. \textbf{Consensus nodes} at the IoT edge maintain the State Chain, run leader election and voting, and store only compact commitments, while \textbf{storage nodes} in the backhaul/cloud store encoded symbols of the Data Chain and serve them on demand. Gateways in Fig.~\ref{fig:system_architecture} belong to the \emph{storage plane} for data dissemination and backhaul and are not part of the consensus plane nor its broadcast assumptions.

\begin{figure}[t]
	\centering
    \resizebox{1\columnwidth}{!}{
	\begin{tikzpicture}[
		scale=0.9, transform shape,
		font=\sffamily,
		node distance=0.2cm,
		outer_node/.style={rectangle, draw=gray!60, thick, fill=white, rounded corners=4pt, minimum width=4.0cm, minimum height=2.6cm, inner sep=5pt, drop shadow={opacity=0.1}},
		style_header/.style={rectangle, draw=ieee_orange!80, fill=ieee_orange!10, thick, minimum width=3.6cm, minimum height=0.8cm, rounded corners=2pt, align=center, font=\footnotesize\bfseries},
		style_body/.style={rectangle, draw=ieee_blue!80, fill=ieee_blue!10, thick, minimum width=3.6cm, minimum height=1.2cm, rounded corners=2pt, align=center, font=\footnotesize},
		style_inner_box/.style={rectangle, draw=ieee_orange!50, fill=white, minimum width=1.1cm, minimum height=0.5cm, rounded corners=1pt, font=\tiny, inner sep=1pt}
		]
		\node[outer_node] (DataBlockContainer) at (0,0) {};
		\node[style_header] (D_Header) at ($(DataBlockContainer.north)-(0, 0.55)$) {};
		\node[anchor=west, font=\footnotesize\bfseries] at ($(D_Header.west)+(0.1,0)$) {Block Header};
		\node[style_inner_box, anchor=east] at ($(D_Header.east)-(0.1,0)$) {Parent Hash};
		\node[style_body] (D_Body) at ($(DataBlockContainer.south)+(0, 0.75)$) {Transaction Payload\\(Erasure Coded)};
		\node[below=0.1cm of DataBlockContainer, font=\small\bfseries] {Data Block (Storage)};
		
		\node[outer_node] (StateBlockContainer) at (4.5,0) {};
		\node[style_header] (S_Header) at ($(StateBlockContainer.north)-(0, 0.55)$) {};
		\node[anchor=west, font=\footnotesize\bfseries] at ($(S_Header.west)+(0.1,0)$) {Block Header};
		\node[style_inner_box, anchor=east] at ($(S_Header.east)-(0.1,0)$) {Parent Hash};
		\node[style_body] (S_Body) at ($(StateBlockContainer.south)+(0, 0.75)$) {\textbf{State Root} \&\\ \textbf{Data Commitment} $C_p$\\(Merkle Root)};
		\node[below=0.1cm of StateBlockContainer, font=\small\bfseries] {State Block (Consensus)};
	\end{tikzpicture}
    }

	\caption{Proposed dual-chain structure. The \textit{Data Chain} stores erasure-coded transaction payloads, while the \textit{State Chain} maintains consensus metadata and a compact commitment $C_p$ (Merkle root) that anchors payload integrity.}
	\label{fig:block_structure}

\end{figure}
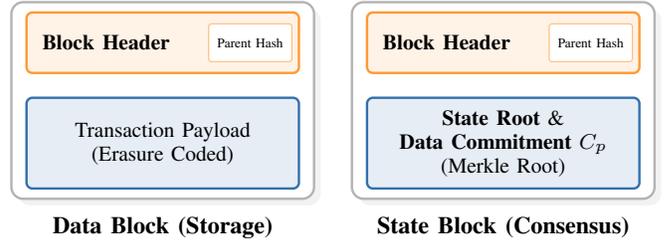

\begin{figure}[t]
	\centering
	\resizebox{1\columnwidth}{!}{%
		\begin{tikzpicture}[
			scale=0.9, transform shape,
			node distance=1.5cm,
			iot_node/.style={
				circle,
				draw=ieee_blue,
				thick,
				fill=white,
				minimum size=1.0cm,
				align=center,
				inner sep=1pt,
				font=\scriptsize\bfseries,
				drop shadow={opacity=0.15, shadow xshift=1pt, shadow yshift=-1pt}
			},
			server_node/.style={
				rectangle,
				draw=black!60,
				fill=gray!5,
				thick,
				minimum width=0.8cm,
				minimum height=1.0cm,
				rounded corners=1pt,
				drop shadow={opacity=0.2}
			},
			tower_node/.style={
				isosceles triangle,
				isosceles triangle apex angle=40,
				draw=black!70,
				fill=white,
				thick,
				shape border rotate=90,
				minimum height=1.4cm,
				minimum width=0.9cm,
				path picture={
					\foreach \y in {0.2, 0.5, 0.8} {
						\draw[gray!50, thin]
						($(path picture bounding box.south west)!\y!(path picture bounding box.north)$) --
						($(path picture bounding box.south east)!\y!(path picture bounding box.north)$);
					}
					\draw[gray!50, thin] (path picture bounding box.south) -- (path picture bounding box.north);
				}
			},
			wireless_link/.style={
				draw=ieee_blue!80,
				dashdotted,
				thick
			},
		wireless_peer/.style={
			draw=ieee_blue!35,
			dashed,
			line width=0.6pt
		},
			wired_link/.style={
				draw=red!70,
				thick
			}
			]
			
			\begin{pgfonlayer}{backgrounds}
				\fill[ieee_blue!6] (3.5,4.9) ellipse [x radius=4.6cm, y radius=1.55cm];
				\draw[ieee_blue!35, dashed, thick] (3.5,4.9) ellipse [x radius=4.6cm, y radius=1.55cm];
			\end{pgfonlayer}
			\node[font=\scriptsize\sf, text=ieee_blue!70] at (3.5,6.3) {Single-hop broadcast domain (mutual overhearing)};
			
			\node[tower_node] (Gateway1) at (1.5, 2.5) {};
			\node[tower_node] (Gateway2) at (5.5, 2.5) {};
			
			\foreach \ap in {Gateway1, Gateway2} {
				\draw[thick, black!70] (\ap.north) -- ++(0, 0.3);
				\foreach \r in {0.2, 0.4} {
					\draw[ieee_blue!50, thick] ($(\ap.north)+(0,0.3)$) arc (90:150:\r);
					\draw[ieee_blue!50, thick] ($(\ap.north)+(0,0.3)$) arc (90:30:\r);
				}
			}
			\node[below=0.1cm of Gateway1, font=\scriptsize\sf] {Gateway 1};
			\node[below=0.1cm of Gateway2, font=\scriptsize\sf] {Gateway 2};
			
			\node[iot_node] (N1) at (0.9, 5.2) {IoT\\Node};
			\node[iot_node] (N2) at (2.2, 5.7) {IoT\\Node};
			\node[iot_node] (N5) at (3.5, 5.1) {IoT\\Node};
			\node[iot_node] (N3) at (4.8, 5.7) {IoT\\Node};
			\node[iot_node] (N4) at (6.1, 5.2) {IoT\\Node};
			\foreach \n in {N1,N2,N3,N4,N5} {
				\draw[ieee_blue!25, thin] (\n) circle [radius=0.55cm];
			}
			
			\node[server_node] (S1) at (0.2, 0.5) {};
			\node[server_node] (S2) at (3.5, 0.0) {};
			\node[server_node] (S3) at (6.8, 0.5) {};
			
			\foreach \s in {S1, S2, S3} {
				\foreach \y in {-0.2, 0.0, 0.2} {
					\draw[gray!30] ($(\s.center) + (-0.25, \y)$) -- ($(\s.center) + (0.25, \y)$);
				}
				\node[below=0.05cm of \s, font=\tiny\sf] {Storage};
			}
			

			
			\draw[wireless_link] (N1) -- ($(Gateway1.north)+(0,0.3)$);
			\draw[wireless_link] (N2) -- ($(Gateway1.north)+(0,0.3)$);
			\draw[wireless_link] (N5) -- ($(Gateway1.north)+(0,0.3)$);
			
			\draw[wireless_link] (N5) -- ($(Gateway2.north)+(0,0.3)$);
			\draw[wireless_link] (N3) -- ($(Gateway2.north)+(0,0.3)$);
			\draw[wireless_link] (N4) -- ($(Gateway2.north)+(0,0.3)$);
			
			\draw[wired_link] (S1) -- (Gateway1.south);
			\draw[wired_link] (S2) -- (Gateway1.south);
			\draw[wired_link] (S2) -- (Gateway2.south);
			\draw[wired_link] (S3) -- (Gateway2.south);
			
			\draw[wired_link] (S1) -- (S2);
			\draw[wired_link] (S2) -- (S3);

			\tikzset{
				state_flow/.style={->, thick, dashed, draw=ieee_orange!85},
				data_flow/.style={->, thick, dashed, draw=ieee_blue!70}
			}
			
			\draw[state_flow]
			(N5) .. controls (3.6,3.8) and (3.0,3.2) .. (Gateway1.north)
			node[midway, right, font=\tiny, text=ieee_orange!85] {State Block (id\_p, C\_p)};
			
			\draw[data_flow]
			(Gateway1.south) .. controls (2.2,1.6) and (1.2,1.2) .. (S2.north)
			node[midway, left, font=\tiny, text=ieee_blue!70] {Encoded symbols q\_i};

			\draw[wireless_peer] (N1) -- (N2);
			\draw[wireless_peer] (N1) -- (N3);
			\draw[wireless_peer] (N1) -- (N4);
			\draw[wireless_peer] (N1) -- (N5);
			
			\draw[wireless_peer] (N2) -- (N3);
			\draw[wireless_peer] (N2) -- (N4);
			\draw[wireless_peer] (N2) -- (N5);
			
			\draw[wireless_peer] (N3) -- (N4);
			\draw[wireless_peer] (N3) -- (N5);
			
			\draw[wireless_peer] (N4) -- (N5);

			\node[anchor=west, font=\scriptsize] at (6.9, 3.2) {\textbf{Legend}};
		    \draw[wireless_peer] (6,2.9) -- (6.7,2.9);
			\node[anchor=west, font=\scriptsize] at (6.8,2.9) {Wireless (iot-iot)};
			\draw[wireless_link] (6,2.6) -- (6.7,2.6);
			\node[anchor=west, font=\scriptsize] at (6.8,2.6) {Wireless (iot-gatway)};
			\draw[wired_link] (6,2.3) -- (6.7,2.3);
			\node[anchor=west, font=\scriptsize] at (6.8,2.3) {Backhaul/wired};
			\draw[state_flow] (6,2.0) -- (6.7,2.0);
			\node[anchor=west, font=\scriptsize] at (6.8,2.0) {State/consensus flow};
			\draw[data_flow] (6,1.7) -- (6.7,1.7);
			\node[anchor=west, font=\scriptsize] at (6.8,1.7) {Data/storage flow};
			
		\end{tikzpicture}
	}

	\caption{System architecture. Consensus nodes (IoT devices) form a \emph{single-hop wireless broadcast cluster} for State Chain agreement, while storage nodes (servers) maintain Data Chain persistence via gateways/backhaul.}
	
	\label{fig:system_architecture}
\end{figure}
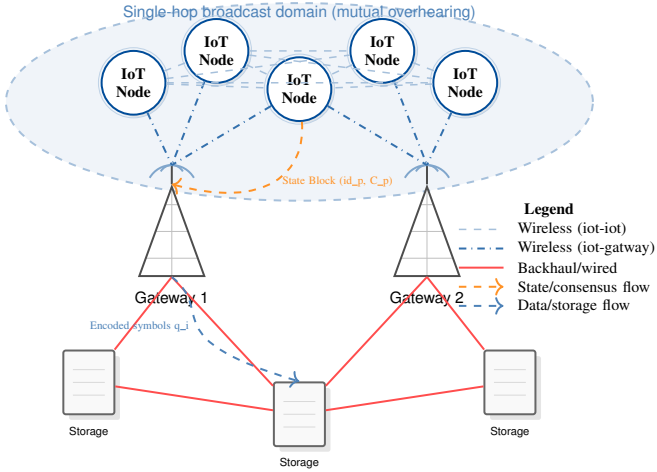

\subsection{Integrity Verification via Hash-Based Commitments}
\label{subsec:integrity_hashing}

Because Raptor codes are not cryptographic, a malicious storage node could return corrupted symbols to disrupt decoding. We therefore add a hash-based integrity mechanism anchored in the State Chain. As explained, we consider two disjoint sets of nodes: consensus nodes and storage nodes. Let $n$ be the number of consensus nodes, among which up to $f$ may be Byzantine; let $s$ be the number of storage nodes provisioned for payload availability, among which up to $f_s$ may be Byzantine or unavailable. For each payload block $p$, the encoder produces $M$ coded symbols (typically $M \ge s$) to be distributed across storage nodes.

For a payload block $p$ (with unique identifier $id_p$, e.g., epoch/height and payload hash), the encoder produces encoded symbols
$Q_p = \{q_1,\dots,q_M\}$.
For each symbol, compute
$h_i = \mathsf{Hash}(id_p \,\|\, i \,\|\, q_i)$,
and define the verification set
$H_p = \{h_1,\dots,h_M\}$.

Embedding $H_p$ directly in the State Block would make the State Chain grow with $M$. Instead, we build a Merkle tree over $\{h_i\}$ and store only its root: $C_p = \mathsf{MerkleRoot}(H_p)$. The consensus layer records $C_p$ in the State Block (Fig.~\ref{fig:block_structure}). During reconstruction, a storage node returns $(q_i, \pi_i)$, where $\pi_i$ is a Merkle inclusion proof for $h_i$. The requester recomputes $h_i$, verifies $\pi_i$ against $C_p$, and accepts $q_i$ only if the proof is valid. This keeps the consensus data compact while providing integrity under standard hash assumptions~\cite{merkle1988digital}. Fig.~\ref{fig:hash_fixed} illustrates the end-to-end commitment generation, from payload encoding to the Merkle root $C_p$ stored in the State Chain.

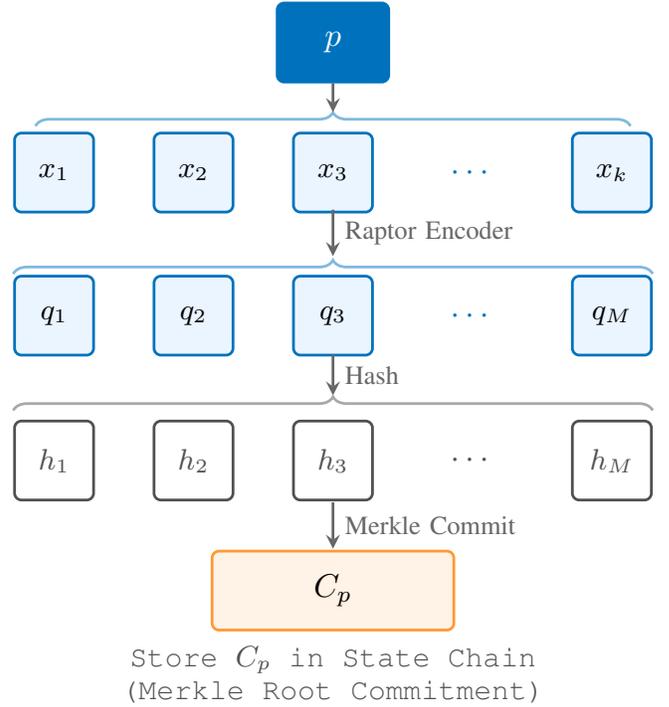
\begin{figure}[t]
	\centering
	\definecolor{tech_blue}{RGB}{0, 114, 189}
	\definecolor{tech_gray}{RGB}{80, 80, 80}
	\definecolor{light_blue}{RGB}{235, 245, 255}
	\usetikzlibrary{calc, decorations.pathreplacing}
	\resizebox{1\columnwidth}{!}{
	\begin{tikzpicture}[
		scale=0.95, transform shape,
		font=\sffamily\small,
		box/.style={
			rectangle,
			minimum size=0.85cm,
			rounded corners=2pt,
			thick,
			align=center
		},
		style_blue/.style={
			box,
			draw=tech_blue,
			fill=light_blue,
			text=black
		},
		style_source/.style={
			box,
			draw=tech_blue,
			fill=tech_blue,
			text=white,
			font=\bfseries
		},
		style_gray/.style={
			box,
			draw=tech_gray,
			fill=white,
			text=tech_gray
		},
		style_commit/.style={
			box,
			draw=ieee_orange!80,
			fill=ieee_orange!10,
			text=black,
			font=\bfseries
		},
		myarrow/.style={
			->,
			thick,
			color=black!60,
			>=stealth
		},
		mybrace/.style={
			decorate,
			decoration={brace, amplitude=4pt},
			thick
		}
		]
		
		\node[style_source, minimum width=1.2cm] (p) at (0, 0) {$p$};
		
		\draw[mybrace, tech_blue!50] (-3.2, -0.9) -- (3.2, -0.9);
		
		\draw[myarrow] (p.south) -- (0, -0.75);
		
		\node[style_blue] (x1) at (-3.0, -1.4) {$x_1$};
		\node[style_blue] (x2) at (-1.5, -1.4) {$x_2$};
		\node[style_blue] (x3) at (0, -1.4)    {$x_3$};
		\node[font=\bfseries, text=tech_blue] at (1.5, -1.4) {$\cdots$};
		\node[style_blue] (xk) at (3.0, -1.4)  {$x_k$};
		
		\draw[myarrow] (0, -1.8) -- (0, -2.3) node[midway, right, font=\footnotesize] {Raptor Encoder};
		
		\node[style_blue] (q1) at (-3.0, -2.94) {$q_1$};
		\node[style_blue] (q2) at (-1.5, -2.94) {$q_2$};
		\node[style_blue] (q3) at (0, -2.94)    {$q_3$};
		\node[font=\bfseries, text=tech_blue] at (1.5, -2.94) {$\cdots$};
		\node[style_blue] (qM) at (3.0, -2.94)  {$q_M$};
		
		\draw[mybrace, tech_blue!50]
		($(q1.north west)+(0, 0.01)$) -- ($(qM.north east)+(0, 0.01)$);
		
		\draw[myarrow] (0, -3.36) -- (0, -3.8) node[midway, right, font=\footnotesize] {Hash};
		
		\node[style_gray] (h1) at (-3.0, -4.5) {$h_1$};
		\node[style_gray] (h2) at (-1.5, -4.5) {$h_2$};
		\node[style_gray] (h3) at (0, -4.5)    {$h_3$};
		\node[font=\bfseries, text=tech_gray] at (1.5, -4.5) {$\cdots$};
		\node[style_gray] (hM) at (3.0, -4.5)  {$h_M$};
		
		\draw[mybrace, tech_gray!50]
		($(h1.north west)+(0, 0.1)$) -- ($(hM.north east)+(0, 0.1)$);
		
		\draw[myarrow] (0, -4.95) -- (0, -5.45) node[midway, right, font=\footnotesize] {Merkle Commit};
		
		\node[style_commit, minimum width=2.6cm, minimum height=0.85cm] (Cp) at (0, -5.9) {$C_p$};
		
		\node[font=\small\ttfamily, color=tech_gray, align=center] at (0, -6.8)
		{Store $C_p$ in State Chain\\(Merkle Root Commitment)};
		
	\end{tikzpicture}
    }

	\caption{Generation of the payload integrity commitment. Encoded symbols are hashed and then committed via a Merkle root $C_p$, which is stored in the State Chain to anchor integrity while keeping consensus blocks lightweight.}
	\label{fig:hash_fixed}

\end{figure}

\subsection{Operational Workflow}
\label{subsec:workflow}

The lifecycle of a block involves: (i) proposing and notarizing a lightweight State Block, (ii) distributing and storing erasure-coded payload symbols, and (iii) on-demand decoding with integrity checks. The workflow is designed to align with Streamlet finality: only \emph{finalized} commitments should trigger irreversible storage actions such as pruning.

\subsubsection{Encoding and Distribution (per epoch)}
This stage erasure-codes the payload, distributes coded symbols to storage nodes, and records the integrity commitment $(id_p, C_p)$ in the State Chain:
\begin{enumerate}
	\item[i)] \textbf{Block formation:} The leader forms a payload block $p$ that contains transactions and constructs the corresponding State Block that contains $(id_p, C_p)$ and other consensus metadata.
	\item[ii)]  \textbf{Erasure coding:} Split $p$ into $k$ source symbols and generate encoded symbols using Raptor coding \cite{shokrollahi2006raptor}. Choose an overhead parameter $\epsilon>0$ so that successful decoding occurs with high probability from any set of at least $k(1+\epsilon)$ valid symbols.
	\item[iii)]  \textbf{Commitment anchoring:} Compute $C_p$ as the Merkle root over symbol hashes (Section~\ref{subsec:integrity_hashing}) and include $(id_p, C_p)$ in the State Block broadcast for consensus.
	\item[iv)]  \textbf{Storage dispatch:} Distribute encoded symbols to storage nodes via gateways and backhaul. Storage nodes index symbols by $(id_p, i)$.
\end{enumerate}

\textbf{Storage availability condition.}
Recall that among $s$ storage nodes, up to $f_s$ may be Byzantine or unavailable. A sufficient condition for successful decoding is
\begin{equation}
	k(1+\epsilon) \le s - f_s.
	\label{eq:storage_sufficient_condition}
\end{equation}
This separates consensus fault tolerance $(n,f)$ from storage availability $(s,f_s)$; we use this distinction in the reliability and liveness analysis that follows.

\subsubsection{Storage Optimization (Pruning)}
This stage reclaims storage from orphaned payloads while preserving safety under State Chain finality:
\begin{enumerate}
	\item[i)] \textbf{Detection (consensus-driven):} After blocks become \emph{finalized} on the State Chain, any State Blocks that are neither finalized nor ancestors of a finalized block are considered invalid.
	\item[ii)] \textbf{Safe purging rule:} Storage nodes delete symbols only for payloads whose commitments correspond to invalidated State Blocks. Finality is the safety trigger; this prevents deleting data that could later become part of the canonical chain \cite{chan2020streamlet}.
\end{enumerate}

\subsubsection{Secure Decoding (On-demand Retrieval)}
This stage retrieves coded symbols on demand, verifies each symbol against the committed root $C_p$, and decodes the payload once enough valid symbols are collected:
\begin{enumerate}
	\item[i)] It requests symbols for $id_p$ from storage nodes.
	\item[ii)] For each received symbol $(q_i,\pi_i)$, it verifies integrity against the committed root $C_p$ from the finalized State Chain, discarding invalid symbols.
	\item[iii)] It runs the Raptor decoder once at least $k(1+\epsilon)$ valid symbols are collected \cite{shokrollahi2006raptor}.
\end{enumerate}

\section{Theoretical Analysis of Consensus Convergence}
\label{sec:analysis}

This section analyzes the convergence (liveness) of Wireless Streamlet under the packet-erasure wireless model introduced in Section~\ref{sec:system_model} and the protocol design in Section~\ref{sec:protocol_design}. In particular, we derive a conservative lower bound on the per-epoch notarization probability and then obtain the expected finality latency using a Markov-chain argument consistent with Streamlet's three-consecutive-notarized rule.

\subsection{Model Assumptions and Notation}

We work in the weak-synchrony, epoch-based model of Section~\ref{sec:system_model}. The analysis focuses on the \emph{consensus plane} (State Chain). Bulk payload dissemination is decoupled and handled by the storage layer (Section~\ref{sec:encoding_scheme}), hence it does not affect the \emph{notarization} probability except through the availability of compact commitments $(id_p,C_p)$.

\noindent\textbf{A1 (System size and Byzantine faults).}
We adopt the node/adversary model in Section~\ref{sec:system_model}: there are $n$ consensus nodes with at most $f$ Byzantine, satisfying $n \ge 3f+1$. Let $h=n-f$ denote the number of honest nodes.

\noindent\textbf{A2 (TDMA epochs).}
We adopt the TDMA epoch structure in Section~\ref{subsec:tdma_voting}: each epoch contains $(n+1)$ scheduled slots (one proposal slot plus $n$ voting slots), and the epoch duration is given by \eqref{eq:epoch_duration}.

\noindent\textbf{A3 (Packet erasure channel lower bound).}
We adopt the non-degenerate honest connectivity assumption in Section~\ref{sec:system_model}: there exists $p_H>0$ such that any honest-to-honest transmission attempt succeeds with probability at least $p_H$.

\noindent\textbf{A4 (Bounded retransmissions per slot).}
Each scheduled slot permits up to $K_{\mathrm{tx}}$ transmission attempts of the same fixed-size consensus packet. Define the within-slot delivery lower bound
\begin{equation}
	\hat{p} \triangleq 1-(1-p_H)^{K_{\mathrm{tx}}}.
	\label{eq:within_slot_delivery}
\end{equation}

\noindent\textbf{A5 (Honest-leader probability under CALE).}
Let $\pi_e$ denote the probability that the elected leader in epoch $e$ is honest.
Under CALE (Section~\ref{subsec:cale}), all nodes compute a deterministic weight table
$\{w_i(e)\}_{i=1}^n$ from the latest \emph{finalized} State Chain checkpoint, where $w_i(e)$ is derived from verifiable quorum certificates that include receiver-measured, vote-signed CSI tags.The unique leader $L(e)$ is then selected by the deterministic weighted winner rule in \eqref{eq:cale_rule}. Let $\mathcal{H}$ be the set of honest nodes with $|\mathcal{H}|=h$. Assuming the normalized hash $\mathsf{Hash}_{nor}$ behaves as a uniform pseudo-random source across $\mathsf{pk}_i$, the weighted
winner rule induces the selection probability
\begin{equation}
\begin{aligned}
    &\Pr[L(e)=i] = \frac{w_i(e)^{\alpha}}{\sum_{k\in\mathcal{N}} w_k(e)^{\alpha}},\\
	&\pi_e \triangleq \Pr[L(e)\ \text{is honest}]
	= \frac{\sum_{i\in\mathcal{H}} w_i(e)^{\alpha}}{\sum_{k\in\mathcal{N}} w_k(e)^{\alpha}}.
\end{aligned}
	\label{eq:pi_cale}
\end{equation}
For uniform election (all $w_i(e)\equiv 1$), \eqref{eq:pi_cale} reduces to $\pi_e = h/n$. In the liveness analysis, we use a conservative epoch-independent lower bound
$\pi \triangleq \inf_e \pi_e > 0$. Because $w_i(e)$ is computed from finalized quorum certificates (rather than self-reported scores), Byzantine nodes cannot arbitrarily inflate their election weights beyond what is supported by verifiable CSI-tagged votes.

\subsection{Safety Discussion}

Wireless Streamlet preserves Streamlet's safety argument at the \emph{protocol layer}: votes are signed, notarization requires a $2f+1$ quorum, and the ``longest notarized chain'' rule prevents two conflicting finalized blocks under the standard Streamlet assumptions \cite{chan2020streamlet}.

\begin{remark}[Threat model clarification]
We leverage wireless broadcast (Section~\ref{subsec:broadcast_no_echo}) purely for efficiency---to remove the application-layer echo phase. We do not assume the physical layer provides a cryptographic non-equivocation guarantee: a Byzantine leader can still attempt equivocation by transmitting different payloads over time or by exploiting capabilities outside our model (e.g., beamforming). Accordingly, we adopt the same threat model as in Section~\ref{subsec:broadcast_no_echo}---an omni-directional IoT cluster without beamforming-capable adversaries. Under this regime, omitting echo is an engineering choice; its effect is captured by the liveness bounds below rather than by the safety argument.
\end{remark}

\subsection{Liveness: Per-Epoch Notarization Probability}

We define an epoch to be successful if the designated leader proposes a valid block and \emph{some honest node} obtains at least $2f+1$ valid votes for that proposal by the end of the epoch (i.e., the block becomes notarized within the epoch).

Because we removed explicit echoing, notarization must be achieved by direct overhearing of the scheduled vote broadcasts. We derive a conservative lower bound by analyzing a particular honest observer: the (honest) leader itself. If the honest leader can collect $2f+1$ votes, then a notarization certificate can be formed for the epoch.

\subsubsection{Step 1 (Proposal reception among honest nodes)}

We condition on the leader being honest, which happens with probability $\pi$. During Slot 0, the leader broadcasts the proposal, in the form of a compact State Block header containing parent pointers and at least the pair $(id_p, C_p)$. Let $X$ be the number of honest nodes that successfully receive the proposal in Slot~0. Under the independent-loss lower bound (Section~\ref{sec:system_model}), we have
\begin{equation}
	X \succeq \mathrm{Binomial}(h,\hat{p}),
	\label{eq:proposal_reception_binom}
\end{equation}
where $\succeq$ denotes first-order stochastic dominance, i.e., for all thresholds $t$, $\Pr[X \ge t] \ge \Pr\!\big[\mathrm{Binomial}(h,\hat{p}) \ge t\big]$. Hence a conservative lower bound on the probability that at least $2f+1$ honest nodes are \emph{eligible to vote} is
\begin{equation}
	P_{\mathrm{prop}} \triangleq \Pr[X \ge 2f+1]
	\ge \sum_{x=2f+1}^{h} \binom{h}{x}\hat{p}^{x}(1-\hat{p})^{h-x}.
	\label{eq:P_prop}
\end{equation}

\subsubsection{Step 2 (Vote reception at the leader)}

Only honest nodes that received the proposal will broadcast a vote in their dedicated TDMA slots. We now condition on $X=x$. Each of these $x$ honest voters broadcasts once, allowing up to $K_{tx}$ attempts within its slot. Let $Z$ be the number of those votes that are successfully received by the (honest) leader by the end of the epoch. Using the honest-link lower bound again, we have
\begin{equation}
	Z \mid (X=x) \succeq \mathrm{Binomial}(x,\hat{p}).
	\label{eq:vote_reception_binom}
\end{equation}
Therefore, conditional on $X=x$, a conservative bound on leader-side quorum formation is
\begin{equation}
	\Pr[Z \ge 2f+1 \mid X=x]
	\ge \sum_{z=2f+1}^{x} \binom{x}{z}\hat{p}^{z}(1-\hat{p})^{x-z}.
	\label{eq:P_vote_given_x}
\end{equation}

\subsubsection{Step 3 (Per-epoch success probability)}

Combining \eqref{eq:proposal_reception_binom}--\eqref{eq:P_vote_given_x}, the per-epoch notarization probability admits the following conservative lower bound:
\begin{align}
	q & \triangleq \Pr[\text{epoch successful}]
	\ge \pi \sum_{x=2f+1}^{h} \Pr[X=x]\;\psi(x) \nonumber\\
	&\ge \pi \sum_{x=2f+1}^{h}
	\binom{h}{x}\hat{p}^{x}(1-\hat{p})^{h-x}\;\psi(x),
	\label{eq:q_lower_bound}
\end{align}
where $\psi(x) \triangleq \Pr\!\left[ \mathrm{Binomial}(x,\hat{p}) \ge 2f+1 \right] = \sum_{z=2f+1}^{x}\binom{x}{z}\hat{p}^{z}(1-\hat{p})^{x-z} $. This lower bound makes the ``two-hop'' nature of the epoch explicit: (i) enough honest nodes must hear the proposal to generate votes, and (ii) enough of those votes must be heard by the leader to reach $2f+1$.


\begin{remark}[Heterogeneity vs. conservative bounds]
	The analysis uses a uniform honest-link lower bound $p_H$ to obtain a conservative guarantee. Channel heterogeneity (e.g., some nodes are in deep fading) is addressed in the protocol via CALE: by biasing the unique-leader selection toward nodes with larger \emph{verifiable} connectivity weights (derived from CSI tags carried in signed votes and aggregated from quorum certificates), CALE improves the \emph{empirical} probability that proposals and votes are heard by a quorum compared to uniform random election under the same heterogeneous channel conditions. In our liveness bound, this effect is abstracted into the per-epoch success probability $q$ (and we only require a conservative $\pi>0$ for the probability that the epoch leader is honest).
\end{remark}

\subsection{Convergence to Finality: Expected Epochs and Expected Time}

Streamlet finality requires three consecutive notarized blocks in three consecutive epochs~\cite{chan2020streamlet}. To quantify \emph{finality progress}---i.e., how close the protocol is to forming such a length-3 notarized run---we model the evolution across epochs as a Markov chain. Specifically, we define states $S_k$ for $k\in\{0,1,2,3\}$, where $k$ is the length of the current suffix of consecutive \emph{successful} epochs in which a notarized block is produced. State $S_3$ is absorbing, corresponding to finality being reached.

Let $E_k$ be the expected number of epochs to reach $S_3$ starting from $S_k$. Let $q$ denote the per-epoch success probability, i.e., the probability that an epoch is successful (a notarized block is produced in that epoch). The transitions then satisfy:
\begin{align}
	E_3 &= 0, \label{eq:Ek3}\\
	E_2 &= 1 + (1-q)E_0 + qE_3, \label{eq:Ek2}\\
	E_1 &= 1 + (1-q)E_0 + qE_2, \label{eq:Ek1}\\
	E_0 &= 1 + (1-q)E_0 + qE_1. \label{eq:Ek0}
\end{align}
For $q=1$, the chain deterministically reaches $S_3$ in $3$ epochs, i.e., $E_0=3$. Solving \eqref{eq:Ek3}-\eqref{eq:Ek0} yields the closed-form expectation
\begin{equation}
	E[T_{\mathrm{epochs}}] = E_0 = \frac{1-q^3}{q^3(1-q)}.
	\label{eq:expected_epochs}
\end{equation}
Because each epoch lasts $T_{\mathrm{epoch}}$ seconds (Eq.~\eqref{eq:epoch_duration}), the expected \emph{time-to-finality} is
\begin{equation}
	E[T_{\mathrm{final}}] = T_{\mathrm{epoch}} \cdot E[T_{\mathrm{epochs}}].
	\label{eq:expected_finality_time}
\end{equation}

\begin{theorem}[Liveness under non-degenerate honest connectivity]\label{thm:liveness_positive_q}
	Assume $p_H>0$ and $\pi>0$. Then $\hat{p}>0$ by \eqref{eq:within_slot_delivery}, which implies $q>0$ by \eqref{eq:q_lower_bound}. Consequently, the expected number of epochs to finality in \eqref{eq:expected_epochs} is finite. Moreover, if the per-epoch success events are independent across epochs or, more generally, form a stationary ergodic process with sufficient mixing, then blocks finalize with probability~$1$ as time goes to infinity under the weak synchrony abstraction.
\end{theorem}

\noindent\textbf{Proof sketch.}
The conditions $p_H>0$ and $\pi>0$ yield a strictly positive per-epoch success probability $q>0$. Finality requires three consecutive successful epochs, i.e., the occurrence of a length-3 run. The Markov-chain calculation in \eqref{eq:expected_epochs} gives a finite expectation for the hitting time when $q>0$. For almost-sure finality, independence (or sufficient mixing) ensures that length-3 runs occur infinitely often with probability~$1$, implying that the absorbing state $S_3$ is reached almost surely.

\subsection{Design Parameters and Optimization}

Under our model, the liveness bound suggests two dominant levers: (i) the probability that the elected leader is both honest and well-connected, which is captured by CALE through $\pi$, and (ii) per-slot delivery reliability controlled by bounded retransmissions.

Increasing $K_{\mathrm{tx}}$ raises $\hat{p}$ exponentially fast via \eqref{eq:within_slot_delivery}, which in turn improves $q$. The cost is additional airtime within each TDMA slot. We therefore consider minimizing the expected number of on-air transmission attempts per finalized block.

Let $C_{\mathrm{epoch}}$ denote the expected number of consensus-plane transmission attempts per epoch. Under TDMA there are $(n+1)$ scheduled senders; allowing up to $K_{\mathrm{tx}}$ attempts each yields the conservative proxy
\begin{equation}
	C_{\mathrm{epoch}} \propto (n+1)K_{\mathrm{tx}},
	\label{eq:epoch_cost_proxy}
\end{equation}
which upper-bounds the number of on-air transmission attempts in the consensus plane per epoch. We can then choose $K_{\mathrm{tx}}$ and CALE parameters that affect $\pi$ to minimize the expected cost per finalized block:
\begin{equation}
	\min_{K_{\mathrm{tx}},\,\text{(CALE params)}}\;
	C_{\mathrm{epoch}}\cdot E[T_{\mathrm{epochs}}]\big(q(\hat{p},\pi)\big),
	\label{eq:consensus_cost_optimization}
\end{equation}
where $\hat{p}$ depends on $K_{\mathrm{tx}}$ through \eqref{eq:within_slot_delivery} and $q$ is lower bounded by \eqref{eq:q_lower_bound}.

Storage-plane parameters (e.g., the Raptor overhead $\epsilon$ and the availability condition in \eqref{eq:storage_sufficient_condition}) govern payload retrievability and storage cost, but do not directly enter the notarization probability of the State Chain. This separation is by design in the dual-chain architecture (Section~\ref{subsec:dual_chain_commitment}).

\section{Experimental Evaluation}
\label{sec:evaluation}

This section evaluates the proposed Wireless Streamlet protocol in a wireless IoT setting. 
Our evaluation follows the architecture and assumptions in Sections~\ref{sec:system_model}--\ref{sec:protocol_design}: 
the \emph{consensus plane} (State Chain) disseminates compact State Blocks (headers/commitments), while bulk payloads are handled by the \emph{storage plane} (Data Chain) with erasure coding and verifiable commitments.
Accordingly, we separate results into three parts:
(i) validating the two core design components namely CALE for consensus liveness and the Raptor-coded data plane for payload availability,
(ii) benchmarking end-to-end consensus performance against PBFT and HotStuff under \emph{matched} networking assumptions, and
(iii) quantifying resource efficiency and scalability in terms of storage and bootstrapping.

\subsection{Experimental Setup}
\label{subsec:exp_setup}

We implement Wireless Streamlet on the \textit{Bamboo} benchmarking framework \cite{dong2024bamboo} with a customizable consensus module and a wireless link layer. The default parameters adopted in the experiments are given in Table \ref{tab:exp_params}, unless otherwise stated. The 20\,KB State header includes aggregated metadata/commitments and is intentionally conservative; votes are 1\,KB. Since consensus operates only on headers---while payloads remain off-chain and are referenced via commitments---the header budget must accommodate the necessary aggregated metadata. Results are averaged over 20 independent runs with 95\% confidence intervals. 

\subsubsection{Wireless Network and MAC}
We simulate a wireless ad-hoc IoT cluster with probabilistic packet delivery as in Section~\ref{sec:system_model}. For liveness-oriented experiments, we parameterize the network by an honest-link success lower bound $p_H$ and a bounded retransmission budget $K_{\mathrm{tx}}$ per TDMA slot, yielding the within-slot delivery bound $\hat{p} = 1-(1-p_H)^{K_{\mathrm{tx}}}$. We adopt the epoch-based TDMA voting schedule in Section~\ref{subsec:tdma_voting}, with $(n+1)$ scheduled slots per epoch. The epoch duration follows Eq.~\eqref{eq:epoch_duration}.


\subsubsection{Consensus-Plane Message Sizes (State Chain)}
To match the dual-chain design (Section~\ref{subsec:dual_chain_commitment}), the consensus plane disseminates only compact objects:
\begin{itemize}
	\item \textbf{State Block header:} includes parent pointer(s), epoch/view, and data commitment tuple $(id_p, C_p)$; default size is 20~KB.
	\item \textbf{Vote:} includes a signature, a block reference, and a small receiver-measured CSI tag piggybacked for CALE (cf. \eqref{eq:vote_with_csi}); default size is 1\,KB.
\end{itemize}
Bulk payloads (default 1.2~MB per block) are never broadcast as part of consensus.

\subsubsection{Storage-Plane Model (Data Chain)}
Payload dissemination and retrieval follow Section~\ref{sec:encoding_scheme}. 
We evaluate our proposed Raptor coding by comparing it against the uncoded replication baseline and an optional fixed-rate Reed–Solomon baseline. In the replication baseline, each payload is partitioned into $k$ uncoded fragments and retrieval succeeds only if \emph{all} $k$ original fragments are obtained; each fragment is retransmitted at most $r$ times (or until a global download deadline $T_{\max}$), which explains the sharp success/latency degradation at high PER.
Nodes verify data integrity using Merkle proofs against the commitment $C_p$ anchored in the State Chain.

\subsubsection{Baselines and Fairness Controls}
We compare Wireless Streamlet against PBFT and HotStuff under \emph{matched} conditions:
\begin{itemize}
	\item \textbf{MAC fairness:} Unless explicitly stated, all protocols run over the same TDMA schedule (same $T_{\mathrm{slot}}$, $T_{\mathrm{guard}}$, $K_{\mathrm{tx}}$).
	\item \textbf{Message-size fairness:} All protocols run in a \emph{header-only} consensus mode, i.e., they agree on a header/commitment rather than transmitting the 1.2~MB payload inside consensus.
\end{itemize}

\begin{table}[t]
	\centering
	\caption{Default Parameters in Experimental Evaluation}
    \vspace{-0.3cm}
	\label{tab:exp_params}
	\begin{tabular}{lc}
		\toprule
		\textbf{Parameter} & \textbf{Default Value} \\
		\midrule
		Number of consensus nodes $n$ & 10 \\
		Byzantine bound $f$ & $\lfloor (n-1)/3 \rfloor$ \\
		Slot duration $T_{\mathrm{slot}}$ & 10 ms \\
		Guard time $T_{\mathrm{guard}}$ & 5 ms \\
		Retransmissions per slot $K_{\mathrm{tx}}$ & 2 \\
		Honest-link success lower bound $p_H$ & 0.95 \\
		State header size & 20 KB \\
		Vote size & 1 KB \\
		Payload size (Data Chain) & 1.2 MB \\
		Wireless bandwidth (peak) & 10 Mbps \\
		\bottomrule
	\end{tabular}
\end{table}

\subsection{Validation of Core Mechanisms}
\label{subsec:exp_core}

We first isolate and validate the effectiveness of (i) CALE on consensus liveness and (ii) the coded storage plane on payload availability and retrieval efficiency.

\subsubsection{CALE Robustness under Channel Heterogeneity}
We model a heterogeneous cluster with two link-quality classes: a fraction $\beta$ of nodes experience ``deep fading'' (i.e., their outgoing link success probability drops to $p_{\mathrm{fade}}=0.4$), while the remaining $1-\beta$ nodes retain $p_{\mathrm{good}}=0.8$. We vary $\beta$ from $0$ to $0.5$ and measure the per-epoch notarization rate, which is the probability that a notarized State Block is produced in an epoch. We simulate a single-hop TDMA epoch with one proposal slot and $n$ vote slots.
A directed link from sender $i$ to receiver $j$ succeeds per attempt with probability $p_{i\to j}=1-\mathrm{PER}_{i\to j}$, and with bounded retransmissions $K_{\mathrm{tx}}$ the within-slot delivery probability is
$\hat{p}_{i\to j}=1-(1-p_{i\to j})^{K_{\mathrm{tx}}}$. In each epoch, the leader broadcasts the proposal once (within-slot attempts bounded by $K_{\mathrm{tx}}=2$); every node that receives the proposal sends one vote in its TDMA slot; the epoch is marked notarized if the leader receives at least $2f+1$ valid votes for that proposal by the end of the epoch. Unless otherwise stated, we use the default parameters in Table~\ref{tab:exp_params}.

Fig.~\ref{fig:cale_notarization} shows that under uniform random leader election, the notarization rate decreases as $\beta$ increases, because the leader falls into the deep-fading subset with probability $\beta$, and such leaders are less likely to disseminate proposals and collect sufficient votes for notarization within the same epoch. In contrast, CALE maintains a high notarization rate by deterministically selecting a unique weighted leader using the verifiable weight table derived from the finalized State Chain (Section~\ref{subsec:cale_unique_leader}). Concretely, each vote piggybacks a receiver-measured CSI tag for the leader's proposal (Eq.~\eqref{eq:vote_with_csi}); from the notarization certificate, nodes can recompute the leader's connectivity score via a Byzantine-robust QC-median aggregator (Eq.~\eqref{eq:connectivity_score}) and update the corresponding leader weight for subsequent epochs. The oracle curve corresponds to an idealized policy that always selects the best-connected leader under the same channel model, and thus serves as an empirical upper bound; CALE tracks this bound closely across heterogeneity levels.

\subsubsection{Storage-Plane Resilience via Erasure Coding (Data Availability)}
We next evaluate the impact of erasure coding on the storage plane, focusing on payload \emph{availability} and \emph{retrieval efficiency}. We fix $n=10$ consensus nodes and vary PER from 0\% to 40\%. 
 A payload of size $B_p$ bytes is split into $k$ source symbols of size $B_{\mathrm{sym}}$ bytes, where $k=\lceil B_p/B_{\mathrm{sym}}\rceil$.
Raptor coding generates $M$ encoded symbols; decoding succeeds once $k_{\mathrm{req}}=\lceil k(1+\epsilon)\rceil$ \emph{verified} symbols are collected (overhead $\epsilon>0$). For a fair comparison, the replication baseline corresponds to \emph{uncoded} storage ($\epsilon=0$): the requester must collect and verify \emph{all} $k$ original symbols within the same retry budget $r$, request parallelism $c$, and wall-clock timeout $T_{\max}$. Each symbol is returned together with its Merkle proof; the transmitted object size is $B_{\mathrm{obj}} = B_{\mathrm{sym}} + B_{\mathrm{proof}}$, with $B_{\mathrm{proof}}\approx 32\cdot\lceil\log_2 M\rceil$ bytes for a 256-bit hash. A symbol transmission attempt over the storage link succeeds with probability $1-\mathrm{PER}$; each symbol is retried at most $r$ times and the overall retrieval terminates unsuccessfully if the wall-clock time exceeds $T_{\max}$. A payload is considered successfully retrieved if a node can reconstruct the block and verify it against the commitment $C_p$ within a fixed timeout. We allow at most $c$ concurrent outstanding symbol requests (request parallelism). In the replication baseline, the payload is treated as $k$ uncoded source symbols and retrieval succeeds only if \emph{all} $k$ original symbols are obtained (each verified against $C_p$), under the same per-symbol retry budget $r$ and global timeout $T_{\max}$. Unless otherwise specified, we set $B_{\mathrm{sym}}=200$~KB (thus $k=\lceil B_p/B_{\mathrm{sym}}\rceil=6$ for $B_p=1.2$~MB), $\epsilon=0.1$ (thus $k_{\mathrm{req}}=7$), $M=s=10$, $r=2$, $c=4$, and $T_{\max}=6$~s. With $M=10$, the Merkle proof size is $B_{\mathrm{proof}}\approx 32\cdot\lceil\log_2 M\rceil\approx 128$~B.

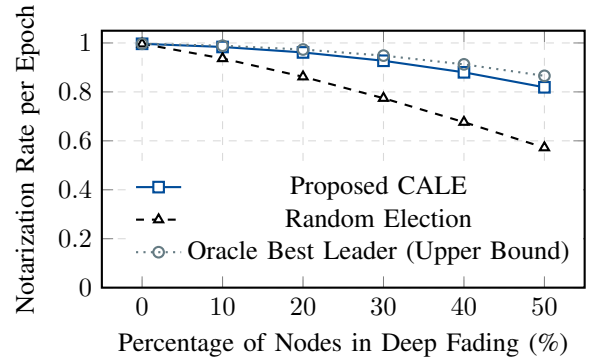
\begin{figure}[t]
	\centering
	\begin{tikzpicture}
		\begin{axis}[
			width=8cm, height=5cm,
			xlabel={Percentage of Nodes in Deep Fading (\%)},
			ylabel={Notarization Rate per Epoch},
			ymin=0, ymax=1.05,
			xtick={0, 10, 20, 30, 40, 50},
			ytick={0, 0.2, 0.4, 0.6, 0.8, 1.0},
			grid=major,
			grid style={dashed, gray!30},
			legend style={at={(0.05, 0.05)}, anchor=south west, draw=none, fill=white},
			thick
			]
			\addplot[color=ieee_blue, mark=square*, mark options={fill=white,solid}] 
coordinates {
  (0,  0.996550) (10, 0.983606) (20, 0.961249)
  (30, 0.927211) (40, 0.879988) (50, 0.818781)
};
			\addlegendentry{Proposed CALE}
			
			\addplot[color=black, mark=triangle*, mark options={fill=white,solid}, dashed] 
coordinates {
  (0,  0.996550) (10, 0.936517) (20, 0.861950)
  (30, 0.774245) (40, 0.676314) (50, 0.572169)
};
			\addlegendentry{Random Election}
			
			\addplot[color=tech_gray, mark=o, mark options={fill=white,solid}, dotted] 
coordinates {
  (0,  0.996550) (10, 0.988641) (20, 0.973195)
  (30, 0.948239) (40, 0.912596) (50, 0.865675)
};
			\addlegendentry{Oracle Best Leader (Upper Bound)}
		\end{axis}
	\end{tikzpicture}
	\vspace{-0.4cm}
	\caption{Impact of leader election on \emph{per-epoch notarization rate} under channel heterogeneity. ``Deep fading'' nodes have lower link success probability.}
		\vspace{-0.4cm}
	\label{fig:cale_notarization}
\end{figure}

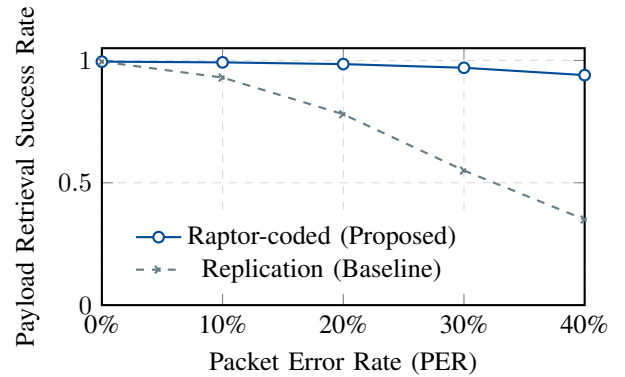
\begin{figure}[t]
	\centering
	\begin{tikzpicture}
		\begin{axis}[
			width=8cm, height=5cm,
			xlabel={Packet Error Rate (PER)},
			ylabel={Payload Retrieval Success Rate},
			xmin=0, xmax=0.4,
			xtick={0, 0.1, 0.2, 0.3, 0.4},
			xticklabels={0\%, 10\%, 20\%, 30\%, 40\%},
			ymin=0, ymax=1.05,
			grid=major,
			grid style={dashed, gray!30},
			legend style={at={(0.05, 0.05)}, anchor=south west, draw=none, fill=white},
			thick
			]
			\addplot[color=ieee_blue, mark=*, mark options={fill=white}] 
			coordinates {
				(0, 0.995) (0.1, 0.992) (0.2, 0.985) (0.3, 0.970) (0.4, 0.940)
			};
			\addlegendentry{Raptor-coded (Proposed)}
			
			\addplot[color=tech_gray, mark=x, dashed] 
			coordinates {
				(0, 0.995) (0.1, 0.930) (0.2, 0.780) (0.3, 0.550) (0.4, 0.350)
			};
			\addlegendentry{Replication (Baseline)}
		\end{axis}
	\end{tikzpicture}
	\vspace{-0.4cm}
	\caption{Storage-plane robustness: payload retrieval success rate vs. PER. Consensus remains header-only; coding improves \emph{data availability}, not notarization.}
		\vspace{-0.4cm}
	\label{fig:payload_availability}
\end{figure}

As shown in Fig.~\ref{fig:payload_availability}, Raptor coding achieves a consistently high retrieval success rate even with substantial packet loss, whereas replication drops rapidly as PER increases. We define the payload retrieval latency as the end-to-end wall-clock time from issuing the first request for block $id_p$ to successfully verifying each accepted symbol via its Merkle inclusion proof (against $C_p$) and completing payload reconstruction/decoding; this metric includes wireless transmission time and all retransmissions/backoff delays. Fig.~\ref{fig:payload_latency} shows a similar trend in latency: coding can tolerate missing fragments and completes decoding once a sufficient number of symbols are received, which reduces reliance on repeated full-block retransmissions. Overall, coding mainly improves \emph{data availability} rather than consensus liveness, which aligns with the separation between the two chains.

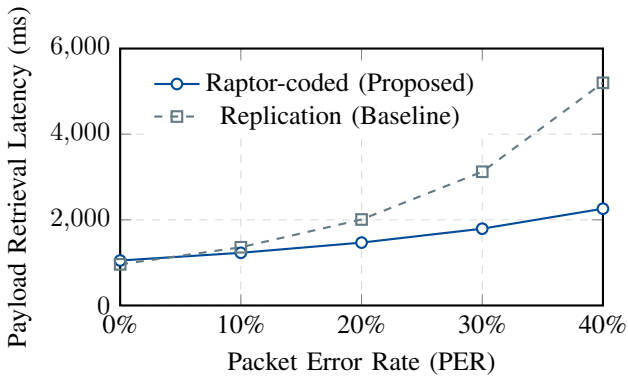
\begin{figure}[t]
	\centering
	\begin{tikzpicture}
		\begin{axis}[
			width=8cm, height=5cm,
			xlabel={Packet Error Rate (PER)},
			ylabel={Payload Retrieval Latency (ms)},
			xmin=0, xmax=0.4,
			xtick={0, 0.1, 0.2, 0.3, 0.4},
			xticklabels={0\%, 10\%, 20\%, 30\%, 40\%},
			ymin=0, ymax=6000,
			grid=major,
			grid style={dashed, gray!30},
			legend style={at={(0.05, 0.95)}, anchor=north west, draw=none, fill=white},
			thick
			]
			\addplot[color=ieee_blue, mark=*, mark options={fill=white}] 
coordinates {
	(0,1050) (0.1,1230) (0.2,1468) (0.3,1793) (0.4,2260)
};
			\addlegendentry{Raptor-coded (Proposed)}
			
			\addplot[color=tech_gray, mark=square, dashed, mark options={fill=white, solid}] 
coordinates {
	(0,960) (0.1,1360) (0.2,2008) (0.3,3123) (0.4,5200)
};
			\addlegendentry{Replication (Baseline)}
		\end{axis}
	\end{tikzpicture}
	\vspace{-0.4cm}
	\caption{Storage-plane efficiency: payload retrieval latency vs. PER. Raptor coding reduces retries by allowing decoding from a subset of received symbols.}
		\vspace{-0.4cm}
	\label{fig:payload_latency}
\end{figure}

\begin{table}[t]
	\centering
	\caption{Wired vs. Wireless (Header-only) for $n=4$ Nodes}
    \vspace{-0.3cm}
	\label{tab:env_comparison_new}
	\begin{tabular}{lcc}
		\toprule
		\textbf{Metric} & \textbf{Wired (Reliable)} & \textbf{Wireless (TDMA)} \\
		\midrule
		Finality latency (avg) & 180 ms & 620 ms \\
		Finality latency (p95) & 230 ms & 980 ms \\
		Notarization rate / epoch & 0.99 & 0.92 \\
		\bottomrule
	\end{tabular}
    \vspace{-0.4cm}
\end{table}

\subsection{End-to-End Consensus Performance}
\label{subsec:exp_consensus}

We now benchmark the consensus plane under wireless constraints and compare against PBFT and HotStuff.

\subsubsection{Wired vs. Wireless Impact (Header-only Consensus)}
We first quantify the performance gap between an ideal reliable network (wired) and the wireless TDMA model.

We report finality latency (three consecutive notarized blocks) and notarization rate per epoch. 
This isolates the overhead from wireless losses and retransmissions without conflating payload dissemination. 
Table~\ref{tab:env_comparison_new} summarizes the impact of the wireless medium on the \emph{consensus plane} under a header-only workload, and isolates the wired--wireless gap using a smaller cluster ($n=4$) to provide a clean, controlled comparison of MAC/medium effects. Compared to the wired reliable setting, the wireless TDMA setting increases the average finality latency from 180 ms to 620 ms, and the tail latency (p95) from 230 ms to 980 ms. Meanwhile, the per-epoch notarization rate drops from 0.99 to 0.92. These results are consistent with our liveness analysis: in a lossy wireless environment, even when each packet is small (header and votes), the protocol still experiences occasional vote losses and retransmissions, which reduce the probability of collecting $2f+1$ votes within a single epoch and consequently inflate the number of epochs required to reach finality.

Fig.~\ref{fig:throughput_env_new} further illustrates the throughput dynamics over time.
The wired baseline sustains around 1.5k tx/s, while the wireless TDMA setting stabilizes around 0.5--0.6k tx/s. The throughput gap primarily comes from reduced effective airtime utilization due to retransmissions and the longer effective epoch duration under wireless loss, rather than payload dissemination (which is excluded in this header-only experiment). Together, Table~\ref{tab:env_comparison_new} and Fig.~\ref{fig:throughput_env_new} quantify the non-negligible wireless overhead and motivate our cross-layer design choices, specifically CALE and deterministic TDMA voting, for stabilizing notarization and finality.

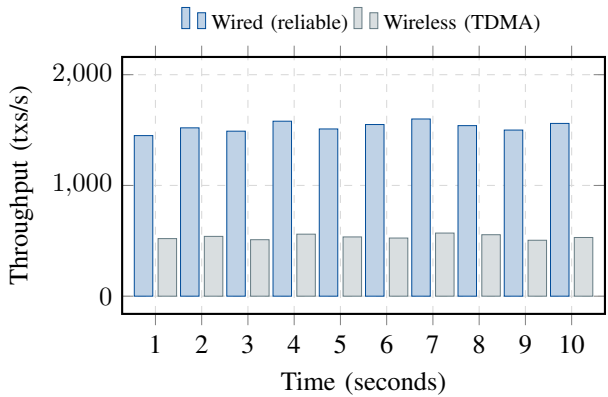
\begin{figure}[t!]
	\centering
	\begin{tikzpicture}
		\begin{axis}[
			width=8cm, height=5cm,
			ybar,
			bar width=7pt,
			enlargelimits=0.08,
			xlabel={Time (seconds)},
			ylabel={Throughput (txs/s)},
			ymin=0, ymax=2000,
			xtick=data,
			symbolic x coords={1,2,3,4,5,6,7,8,9,10},
			grid=major,
			grid style={dashed, gray!30},
			legend style={
				at={(0.5, 1.05)}, anchor=south, legend columns=2,
				draw=none, fill=none, font=\footnotesize
			},
			thick,
			plot coordinates/math parser=false
			]
			\addplot[
			fill=ieee_blue!25,
			draw=ieee_blue,
			line width=0.35pt
			] coordinates {
				(1,1450) (2,1520) (3,1490) (4,1580) (5,1510)
				(6,1550) (7,1600) (8,1540) (9,1500) (10,1560)
			};
			\addlegendentry{Wired (reliable)}
			
			\addplot[
			fill=tech_gray!25,
			draw=tech_gray,
			line width=0.35pt
			] coordinates {
				(1,520) (2,540) (3,510) (4,560) (5,535)
				(6,525) (7,570) (8,555) (9,505) (10,530)
			};
			\addlegendentry{Wireless (TDMA)}
		\end{axis}
	\end{tikzpicture}
	\vspace{-0.4cm}
	\caption{Throughput under wired vs. wireless settings (header-only consensus). Wireless losses reduce throughput due to retransmissions and longer effective epochs.}
		\vspace{-0.4cm}
	\label{fig:throughput_env_new}
\end{figure}

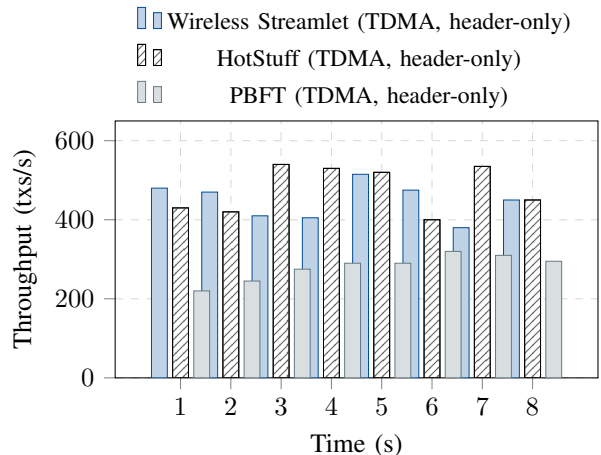
\begin{figure}[t!]  
	\centering
	\begin{tikzpicture}
		\begin{axis}[
			ybar, 
			width=8cm, height=5cm, 
			xlabel={Time (s)}, ylabel={Throughput (txs/s)},
			xmin=0.5, xmax=8.5, 
			ymin=0, ymax=650,
			xtick={1, 2, 3, 4, 5, 6, 7, 8}, 
			grid=major, grid style={dashed, gray!30},
			tick align=outside, tick pos=left,
			bar width=6pt, 
			enlarge x limits=0.1, 
			legend style={
				at={(0.5, 1.01)}, anchor=south,
				legend columns=1,
				draw=none, fill=none,
				font=\small,
				row sep=2pt
			},
			every axis plot/.append style={line width=0.35pt}
			]
			
			\addplot[
			fill=ieee_blue!25, 
			draw=ieee_blue
			] coordinates {
				(1.0, 480) (2.0, 470) (3.0, 410) (4.0, 405) 
				(5.0, 515) (6.0, 475) (7.0, 380) (8.0, 450)
			};
			\addlegendentry{Wireless Streamlet (TDMA, header-only)}			
			\addplot[
			fill=white, 
			draw=black, 
			postaction={pattern=north east lines, pattern color=black!60}
			] coordinates {
				(1.0, 430) (2.0, 420) (3.0, 540) (4.0, 530) 
				(5.0, 520) (6.0, 400) (7.0, 535) (8.0, 450)
			};
			\addlegendentry{HotStuff (TDMA, header-only)}		
			\addplot[
			fill=tech_gray!25, 
			draw=tech_gray
			] coordinates {
				(1.0, 220) (2.0, 245) (3.0, 275) (4.0, 290) 
				(5.0, 290) (6.0, 320) (7.0, 310) (8.0, 295)
			};
			\addlegendentry{PBFT (TDMA, header-only)}		
		\end{axis}
	\end{tikzpicture}
	\vspace{-0.4cm}
	\caption{Throughput comparison under matched TDMA and header-only consensus settings (sampled every 1s).}
		\vspace{-0.4cm}
	\label{fig:throughput_protocols_new}
\end{figure}

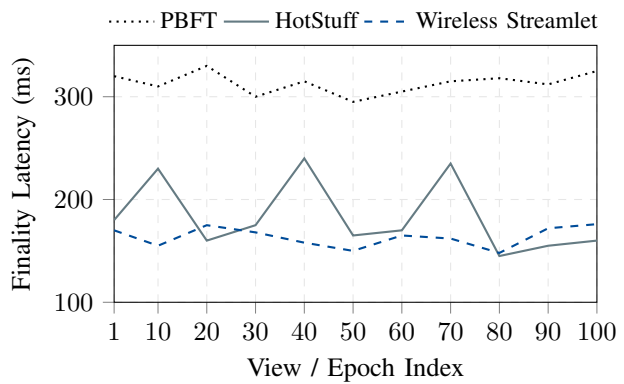
\begin{figure}[t!] 
	\centering
	\begin{tikzpicture}
		\begin{axis}[
			width=8cm, height=5cm,
			xlabel={View / Epoch Index}, ylabel={Finality Latency (ms)},
			xmin=1, xmax=100, ymin=100, ymax=350,
			xtick={1, 10, 20, ..., 100},
			grid=major, grid style={dashed, gray!20},
			tick align=outside, tick pos=left,
			legend style={
				at={(0.5, 1.01)}, anchor=south, legend columns=-1,
				draw=none, fill=none, font=\small
			}
			]
			\addplot[color=black, thick, dotted] coordinates { (1,320)(10,310)(20,330)(30,300)(40,315)(50,295)(60,305)(70,315)(80,318)(90,312)(100,325) };
			\addlegendentry{PBFT}
			\addplot[color=tech_gray, thick] coordinates { (1,180)(10,230)(20,160)(30,175)(40,240)(50,165)(60,170)(70,235)(80,145)(90,155)(100,160) };
			\addlegendentry{HotStuff}
			\addplot[color=ieee_blue, thick, dashed] coordinates { (1,170)(10,155)(20,175)(30,168)(40,158)(50,150)(60,165)(70,162)(80,148)(90,172)(100,176) };
			\addlegendentry{Wireless Streamlet}
		\end{axis}
	\end{tikzpicture}
	\vspace{-0.4cm}
	\caption{Finality latency over 100 views under matched TDMA and header-only settings. Wireless Streamlet shows lower jitter due to deterministic voting slots.}
		\vspace{-0.4cm}
	\label{fig:finality_latency_protocols}
\end{figure}

\subsubsection{Comparison with PBFT and HotStuff (Matched MAC and Message Sizes)}
We compare Wireless Streamlet against PBFT and HotStuff in a 10-node wireless cluster under the same TDMA schedule and header-only message sizes.

Fig.~\ref{fig:throughput_protocols_new} reports real-time throughput. Leader-based protocols (Wireless Streamlet and HotStuff) sustain higher throughput than PBFT under the same airtime/slot budget, because PBFT requires more authenticated exchanges to commit each decision, leaving less effective airtime for transaction progress.

Fig.~\ref{fig:finality_latency_protocols} reports finality latency over 100 views/epochs. Wireless Streamlet exhibits lower variance due to the deterministic TDMA voting structure, which removes random backoff jitter and limits contention.

\subsection{Resource Efficiency and Scalability}
\label{subsec:exp_resource}

Finally, we evaluate the resource efficiency enabled by the dual-chain separation and coded storage.

\subsubsection{Per-node storage reduction} Fig.~\ref{fig:storage_reduction} quantifies the per-node payload storage as the blockchain height increases. Under traditional full replication, each node stores the entire payload stream, causing the per-node burden to grow linearly and quickly reaching 1{,}150{,}000~KB at height 150. In contrast, the proposed coded storage distributes symbols across a storage set, reducing the per-node payload to 230{,}000~KB with 20 storage nodes and further down to 30{,}000~KB with 200 storage nodes at the same height. These results confirm that scaling the storage plane is an effective lever for reducing edge-device storage pressure without sacrificing data availability.

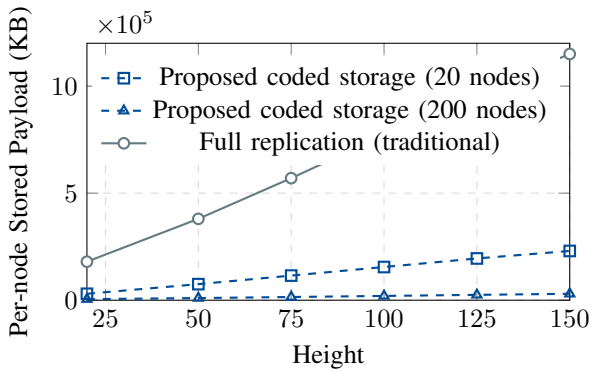
\begin{figure}[t!]
	\centering
	\definecolor{ieee_blue}{RGB}{0, 76, 153}
	\definecolor{tech_gray}{RGB}{101, 123, 131}
	\begin{tikzpicture}
		\begin{axis}[
			width=8cm, height=5cm,
			xlabel={Height}, ylabel={Per-node Stored Payload (KB)},
			xmin=20, xmax=150, ymin=0, ymax=1200000,
			xtick={25, 50, 75, 100, 125, 150},
			scaled y ticks=base 10:-5, ytick scale label code/.code={$\times 10^5$},
			grid=major, grid style={dashed, gray!30},
			legend style={at={(0.01, 0.95)}, anchor=north west, draw=none, fill=white}
			]
			\addplot[color=ieee_blue, dashed, thick, mark=square, mark options={solid, fill=white}] 
			coordinates { (20, 30000) (50, 75000) (75, 115000) (100, 155000) (125, 195000) (150, 230000) };
			\addlegendentry{Proposed coded storage (20 nodes)}
			\addplot[color=ieee_blue, dashed, thick, mark=triangle, mark options={solid, fill=white}] 
			coordinates { (20, 5000) (50, 10000) (75, 15000) (100, 20000) (125, 25000) (150, 30000) };
			\addlegendentry{Proposed coded storage (200 nodes)}
			\addplot[color=tech_gray, thick, mark=*, mark options={rotate=90, fill=white, solid}] 
			coordinates { (20, 180000) (50, 380000) (75, 570000) (100, 770000) (125, 970000) (150, 1150000) };
			\addlegendentry{Full replication (traditional)}
		\end{axis}
	\end{tikzpicture}
	\vspace{-0.4cm}
	\caption{Per-node payload storage vs. blockchain height. Coded storage reduces the \emph{per-node} burden by distributing encoded symbols across more \emph{storage nodes} (participants in the storage plane).}
	\vspace{-0.4cm}
	\label{fig:storage_reduction}
\end{figure}

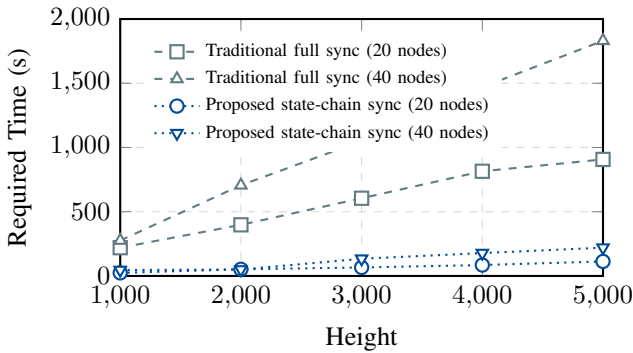
\begin{figure}[t!]
	\centering
	\definecolor{ieee_blue}{RGB}{0, 76, 153}
	\definecolor{tech_gray}{RGB}{101, 123, 131}
	\begin{tikzpicture}
		\begin{axis}[
			width=8cm, height=5cm,
			xlabel={Height}, ylabel={Required Time (s)},     
			xmin=1000, xmax=5000, xtick={1000, 2000, 3000, 4000, 5000},
			ymin=0, ymax=2000,        
			grid=major, grid style={dashed, gray!30}, thick,
			legend style={
				at={(0.05, 0.95)}, anchor=north west, 
				draw=none, fill=white, font=\scriptsize,
				legend columns=1, cells={anchor=west}
			},       
			mark size=2.5pt, mark options={solid, fill=white}
			]
			\addplot[color=tech_gray, dashed, mark=square*] 
			coordinates { (1000, 220) (2000, 398) (3000, 605) (4000, 815) (5000, 907) };
			\addlegendentry{Traditional full sync (20 nodes)}
			\addplot[color=tech_gray, dashed, mark=triangle*] 
			coordinates { (1000, 276) (2000, 706) (3000, 1051) (4000, 1470) (5000, 1831) };
			\addlegendentry{Traditional full sync (40 nodes)}
			\addplot[color=ieee_blue, dotted, mark=*] 
			coordinates { (1000, 24) (2000, 52) (3000, 66) (4000, 85) (5000, 113) };
			\addlegendentry{Proposed state-chain sync (20 nodes)}
			\addplot[color=ieee_blue, dotted, mark=triangle*, mark options={rotate=180, solid, fill=white}] 
			coordinates { (1000, 46) (2000, 50) (3000, 134) (4000, 178) (5000, 220) };
			\addlegendentry{Proposed state-chain sync (40 nodes)}
		\end{axis}
	\end{tikzpicture}
	\vspace{-0.4cm}
	\caption{Bootstrapping time to join the network. The proposed design enables fast synchronization by downloading only the lightweight State Chain; payloads are retrieved on demand from the coded storage plane.}
	\label{fig:joining_time}
	\vspace{-0.4cm}
\end{figure}

\subsubsection{Fast bootstrapping via state-chain-first synchronization}
Fig.~\ref{fig:joining_time} shows the time required for a new node to join and synchronize to a given height. Traditional full synchronization becomes increasingly costly as the height grows, reaching 1831~s at height 5000 with 40 nodes. By downloading only the lightweight State Chain first and retrieving payloads on demand from the coded storage plane, the proposed design reduces the join time to 220~s under the same setting. This separation between header-only consensus and payload storage substantially improves scalability in dynamic wireless deployments where churn is common.

\section{Conclusion}\label{sec:con}
In this paper, we presented \textit{Wireless Streamlet}, a spectrum-aware and cognitive consensus framework tailored for dynamic wireless edge networks. By embedding a verifiable, channel-aware CALE cognitive engine into the consensus logic, we showed how receiver-measured physical-layer channel information (CSI) can guide leader selection to improve block dissemination reliability under fading and interference. The integration of deterministic TDMA scheduling with the single-hop broadcast medium ensures predictable spectral access and linear slot complexity, achieving significantly higher throughput and stability than traditional spectrum-ignorant BFT baselines. To further address the resource constraints of IoT devices, we proposed a coded dual-chain architecture that decouples header-only consensus (State Chain) from payload persistence (Data Chain). This design, powered by erasure coding and on-chain integrity commitments, reduces the per-node storage footprint by up to 97\% while minimizing the redundant spectrum usage required for data retrieval. Collectively, these cross-layer innovations provide a robust foundation for securing cognitive IoT networks, enabling trustworthy coordination in dynamic spectrum environments and supporting emerging distributed edge intelligence workloads.

\ifCLASSOPTIONcaptionsoff
  \newpage
\fi

\bibliographystyle{IEEEtran}
\bibliography{ref}

\end{document}